\documentclass[draftcls,onecolumn,12pt]{IEEEtran}
\usepackage{color}
\usepackage{graphicx}
\usepackage{epstopdf}
\usepackage{amsmath}
\usepackage{amssymb}
\usepackage[english]{babel}
\usepackage{cite}
\usepackage{rotfloat}
\usepackage{mathtools}
\usepackage[justification=centering]{caption}
\usepackage{breqn}
\usepackage{amsmath}
\usepackage{makecell}
\usepackage{algorithm,algorithmic}
\usepackage{multirow}
\usepackage{subfigure}
\usepackage{booktabs}
\usepackage{colortbl}
\definecolor{kugray5}{RGB}{224,224,224}


\makeatletter
\newcommand{\ALOOP}[1]{\ALC@it\algorithmicloop\ #1%
	\begin{ALC@loop}}
	\newcommand{\ENDALOOP}{\end{ALC@loop}\ALC@it\algorithmicendloop}

\makeatother

\usepackage{etoolbox}
\let\mybibitem\bibitem
\renewcommand{\bibitem}[1]{%
	\ifstrequal{#1}{nature}
	{\color{blue}\mybibitem{#1}}
	{\color{black}\mybibitem{#1}}%
}

\graphicspath{ {Figures/} }

\newtheorem{theorem}{Theorem}

\newtheorem{remark}{Remark}
\newtheorem{lemma}{Lemma}

\newtheorem{proof}{Proof}

\newcommand{\epr}{\hfill\(\Box\)}

\DeclareCaptionLabelSeparator{periodspace}{.\quad}

\renewcommand{\figurename}{Fig.}
\addto\captionsenglish{\renewcommand{\figurename}{Fig.}}
\newcommand\numberthis{\addtocounter{equation}{1}\tag{\theequation}}

\newcommand{\norm}[1]{\left\lVert#1\right\rVert} % ||.||
\newcommand{\eq}[1]{\begin{align*}#1\end{align*}} % equation
\newcommand{\abs}[1]{\left|#1\right|} % ||

\newcommand{\tr}[1]{\text{trace}\left(#1\right)}

\newcommand{\mean}[1]{\mathbb{E} \left\{#1\right\}}

\newcommand{\setC}{\mathbb{C}}

\newcommand{\mQ}{\textbf{\textit{Q}}}
\newcommand{\mR}{\textbf{\textit{R}}}
\newcommand{\mH}{\textbf{\textit{H}}} 
\newcommand{\mW}{\textbf{\textit{W}}} 

\newcommand{\mI}{\textbf{\textit{I}}}

\newcommand{\mG}{\textbf{\textit{G}}}
\newcommand{\mE}{\textbf{\textit{E}}}

\newcommand{\mT}{\textbf{\textit{T}}}
\newcommand{\mA}{\textbf{\textit{A}}}
\newcommand{\mB}{\textbf{\textit{B}}}
\newcommand{\mU}{\textbf{\textit{U}}}

\newcommand{\vx}{\textbf{\textit{x}}}
\newcommand{\vy}{\textbf{\textit{y}}}

\newcommand{\vu}{\textbf{\textit{u}}}
\newcommand{\vz}{\textbf{\textit{z}}} 
\newcommand{\vh}{\textbf{\textit{h}}}

\newcommand{\vw}{\textbf{\textit{w}}}
 
\newcommand{\va}{\textbf{\textit{a}}}

\newcommand{\vb}{\textbf{\textit{b}}}

\newcommand{\sm}[1]{\sigma_{1}(#1)}

\newcommand{\comM}{{\mathbf{\Psi}}}
\newcommand{\comMs}{{\mathbf{\Psi}}^{\star}}
\newcommand{\um}[1]{\mu_{1}(#1)}
	
\begin{document}	
	\title{Unequally Sub-connected Architecture for Hybrid Beamforming in Massive MIMO Systems}
	\author{Nhan~Thanh Nguyen and 
		Kyungchun~Lee,~\IEEEmembership{Senior Member,~IEEE}
		%\thanks{This research was supported in part by Basic Science Research Program through the National Research Foundation of Korea (NRF) funded by the Ministry of Education (NRF-2019R1A6A1A03032119). This research was also supported in part by Basic Science Research Program through the NRF, funded by the Ministry of Education (NRF-2016R1D1A1B03933122). (Corresponding author: Kyungchun Lee.)}
        \thanks{N. T. Nguyen is with the Department of Electrical and Information Engineering, Seoul National University of Science and Technology,  Seoul 01811, Republic of Korea (e-mail: nhan.nguyen@seoultech.ac.kr).}
\thanks{K. Lee is with the Department of Electrical and Information Engineering and the Research Center for Electrical and Information Technology, Seoul National University of Science and Technology,  Seoul 01811, Republic of Korea (e-mail: kclee@seoultech.ac.kr).}
	}
	\maketitle
	
	\begin{abstract}
		A variety of hybrid analog-digital beamforming architectures have recently been proposed for massive multiple-input multiple-output (MIMO) systems to reduce energy consumption and the cost of implementation. In the analog processing network of these architectures, the practical sub-connected structure requires lower power consumption and hardware complexity than the fully connected structure but cannot fully exploit the beamforming gains, which leads to a loss in overall performance. In this work, we propose a novel unequal sub-connected architecture for hybrid combining at the receiver of a massive MIMO system that employs unequal numbers of antennas in sub-antenna arrays. The optimal design of the proposed architecture is analytically derived, and includes antenna allocation and channel ordering schemes. Simulation results show that an enhancement of up to $10\%$ can be attained in the total achievable rate by unequally assigning antennas to sub-arrays in the sub-connected system at the cost of a marginal increase in power consumption. Furthermore, in order to reduce the computational complexity involved in finding the optimal number of antennas connected to each radio frequency (RF) chain, we propose three low-complexity antenna allocation algorithms. The simulation results show that they can yield a significant reduction in complexity while achieving near-optimal performance.
	\end{abstract}
	
	\begin{IEEEkeywords}
		Hybrid precoding, analog combining, sub-connected architecture, massive MIMO, millimeter wave.
	\end{IEEEkeywords}
	\IEEEpeerreviewmaketitle
	
	\section{Introduction}
	
	In mobile communication, massive multiple-input multiple-output (MIMO) systems, where a base station (BS) is equipped with a large number of antennas, have recently been considered to drastically improve system performance in terms of spectral and energy efficiency \cite{marzetta2010noncooperative, nguyen2017cell, nguyen2018coverage}. In the conventional frequency band, precoding is typically processed only in the digital domain for interference mitigation among spatial substreams, which results in the requirement of a dedicated radio frequency (RF) chain and an analog-to-digital or digital-to-analog converter (ADC/DAC) for each antenna \cite{gao2016energy, niu2017low}. As a result, the cost and power consumption of the transceiver increase approximately proportionally to the number of antennas \cite{sandhu2003analog, gholam2011beamforming}, which can lead to excessive power consumption in massive MIMO systems. 
	
	% Related works
	\subsection{Related works}
	
	The hybrid analog-digital architecture, where signal processing is divided into the RF and baseband domains, is considered a practical transceiver design for massive MIMO systems because it can provide an enhanced tradeoff between the achievable spectral efficiency and power consumption \cite{gao2018low, park2017dynamic, el2014spatially, alkhateeb2014mimo, heath2016overview}.

	A line of research has sought to optimize the tradeoff between the performance and power consumption of hybrid precoding/combining by optimally designing an analog precoding network based on phase shifters and switches \cite{gholam2011beamforming, payami2016hybrid, mendez2016hybrid, bogale2016number, venkateswaran2010analog, gao2018low}. In \cite{gholam2011beamforming}, Gholam et al. present three simplified analog combining architectures that rely on different combinations of phase shifters and variable gain amplifiers. Out of these three architectures, the one based only on phase shifters provides the best bit-error-rate (BER) performance. Gholam et al. have shown that their proposed architecture can reduce overall system cost and power consumption; however, these reductions come at the expense of a signal-to-noise ratio (SNR) loss of at least 2 dB. In \cite{payami2016hybrid}, Payami et al. propose a technique that successively approximates the desired overall analog precoder as a linear combination of practical analog precoders, which employ only a practical number of phase shifters. Although the approach proposed in \cite{payami2016hybrid} can reduce the power consumption on the RF end, improvements in spectral efficiency are seen only when the channel follows the Rayleigh fading model, which is generally impractical in mmWave communications due to limited scatters. The work of \cite{mendez2016hybrid} proposes switch-only architectures in order to reduce complexity and power consumption on an order of $40\%–75\%$ while providing equivalent or better channel estimation performance and spectral efficiency when compared with structures based on phase shifters for which the number of quantization bits is low. However, when there are more than four quantization bits, switch-only architectures suffer from a significant loss in spectral efficiency  compared to the architectures that employ phase shifters. In order to exploit the advantages of both switches and phase shifters, \cite{bogale2016number} proposes an analog architecture that employs both. This architecture profits from an improved tradeoff between performance and power consumption by optimizing the number of RF chains and phase shifters used in the analog part.

	Most past studies have considered fully connected structures for hybrid precoding to achieve full precoding gains \cite{el2014spatially, alkhateeb2014channel, liang2014low, alkhateeb2015limited, lee2015hybrid, kim2015mse, han2015large, gao2016energy}. To further reduce power consumption and hardware complexity, the sub-connected structure, which connects each RF chain to only a subset of antennas, can also be considered \cite{han2015large, gao2016energy}. However, its beamforming gain is only $1/N$ that of the fully connected structure, where $N$ is the number of sub-antenna arrays in the sub-connected structure. In the literature, few studies \cite{he2016energy, gao2016energy, park2017dynamic, li2017hybrid, zhu2016adaptive, chen2007efficient} have focused on improving the performance of the sub-connected structure. In \cite{gao2016energy}, an algorithm called successive interference cancellation (SIC)-based hybrid precoding was proposed to improve the energy efficiency of the sub-connected architecture. In \cite{park2017dynamic}, a dynamic sub-array structure was proposed to dynamically adapt to the spatial channel covariance matrix. In \cite{he2016energy}, an analog precoder was designed by using the idea of interference alignments. This precoding scheme exploits the alternating direction optimization method, where the phase shifter is adjusted using an analytical structure to optimize the analog precoder and combiner. In \cite{li2017hybrid}, a low-complexity hybrid combining design based on virtual path selection was introduced for both fully connected and sub-connected structures. In \cite{zhu2016adaptive}, a switch-based adaptive sub-connected architecture is proposed for hybrid precoding in multiuser massive MIMO systems. Using a switching network, the precoding weights and their position in the precoding matrix are adaptively selected to match the channel entries with the largest amplitude, resulting in a significant improvement in the total achievable rate. In contrast, an adaptive antenna selection scheme is proposed for the transmitter with a limited number of RF chains in \cite{chen2007efficient}. In this scheme, a subset of the transmit antennas are adaptively connected to RF chains using a switching network, whereas the corresponding subset of transmit antennas is selected by low-complexity transmit selection algorithms. All these schemes have been proposed to optimize the hybrid precoder/combiner for the sub-connected structure in which the same number of antennas or a single antenna is assigned to the sub-arrays. This work proposes a novel sub-connected architecture to further improve performance.

	% Contributions
	\subsection{Contributions}
	Most past work in the area considered hybrid precoding schemes at the transmitter of massive MIMO systems with a focus on the fully connected architecture in analog processing. Compared with the fully connected architecture, the sub-connected structure is more advantageous in terms of practical deployment, complexity, and power consumption. However, its application to the receiver in massive MIMO systems has not been extensively investigated. In this work, we propose a novel unequal sub-connected architecture to improve the achievable rate of hybrid combining at the receiver in massive MIMO systems. The main idea of this scheme is to optimize the number of antennas assigned to the sub-antenna arrays based on channel conditions, which is achieved by the theoretical analysis and low-complexity antenna allocation algorithms.
	
	Although some adaptive hybrid beamforming schemes have been introduced in literature \cite{zhu2016adaptive, chen2007efficient}, our proposed scheme has a novel system structure. Specifically, in most of the prior works on sub-connected architectures, the same number of antennas are assigned to sub-arrays \cite{zhu2016adaptive, gao2016energy, park2017dynamic} or only a subset of the antennas are selected for signal precoding/combining \cite{chen2007efficient}. Unlike those prior works, our proposed architecture allows forming sub-arrays with different numbers of antennas, and all the antennas are employed for signal combining. Further comparison in terms of system structure and performance are given in Section \ref{sec:simulation_result}.

	Our main contributions can be summarized as follows:
	\begin{itemize}
		\item We first investigate a system employing the conventional sub-connected structure, where the same number of antennas are assigned to each sub-antenna array. In particular, we derive the upper bound of its total achievable rate when factorization-aided analog combining is employed. We then show that this upper bound is unreachable owing to the fixed allocation of antennas in the conventional sub-connected structure, which limits overall system performance.
		
		\item We then propose a novel sub-connected architecture called the \textit{unequal sub-array} (UESA) architecture, where different numbers of antennas can be assigned to sub-antenna arrays based on channel conditions. In the proposed architecture, the upper bound of the total achievable rate can be enhanced and, at the same time, the total achievable rate becomes more likely to reach its upper bound. As a result, overall performance can be improved.
		
		\item Although an improvement in performance is achieved, the proposed UESA architecture requires high computational complexity to find the optimal number of antennas for each sub-antenna array via an exhaustive search (ES). To solve this problem, we propose three low-complexity near-optimal algorithms that can substantially reduce the computational burden at the cost of only marginal performance losses.
		
		\item The power consumption and energy efficiency of the proposed architecture are also evaluated. Although the switching network in the proposed architecture causes additional power consumption compared with the conventional architecture, this is not significant. Finally, simulations are performed to justify the performance improvements of the proposed architecture and algorithms. Furthermore, the performance of the proposed schemes are compared to those of existing adaptive hybrid beamforming architectures, including the schemes in \cite{zhu2016adaptive} and \cite{chen2007efficient}. Our simulation results show that the proposed schemes perform far better than the adaptive antenna selection scheme in \cite{chen2007efficient}. It is also shown that by combining the proposed schemes with the adaptive hybrid beamforming scheme in \cite{zhu2016adaptive}, significant improvements in the total achievable rates can be achieved.
	\end{itemize}
	
	% Paper structure
	The remainder of this paper is organized as follows: Section II introduces the system model, and Section III presents the factorization-aided analog-combining algorithm and the achievable rate of the conventional sub-connected structure. In Section IV, the total achievable rate and power consumption of the proposed sub-connected architecture are analyzed, and low-complexity antenna allocation algorithms are proposed. Section V presents simulation results, followed by the conclusion in Section VI.

	\textit{Notations}: Throughout this paper, scalars, vectors, and matrices are denoted by lower-case, bold-face lower-case, and bold-face upper-case letters, respectively. The $(i,j)$th element of a matrix $\textbf{\textit{A}}$ is denoted by $a_{i,j}$, whereas $(\cdot)^T$ and $(\cdot)^H$ denote the transpose and conjugate transpose of a matrix, respectively. Furthermore, $\norm{\cdot}$ represents the Frobenius norm of a matrix while $\abs{\mathbb{S}}$ denotes the cardinality of a set $\mathbb{S}$. Moreover, $\mA(i:j)$ and $\mA(i:j,n)$ with $j>i$ denote a sub-matrix of $\mA$ and a sub-column of the $n$th column of $\mA$, respectively, consisting of the elements on rows $\left\{i, i+1, \ldots, j \right\}$ of $\mA$. Finally, $\mI_N$ denotes the identity matrix of size $N \times N$, and $\mathbf{0}$ is a vector of zeros with an appropriate number of dimensions.
	
	\section{System model}
	
	\begin{figure*}[t]
		\subfigure[Conventional ESA architecture]
		{
			\includegraphics[scale=0.6]{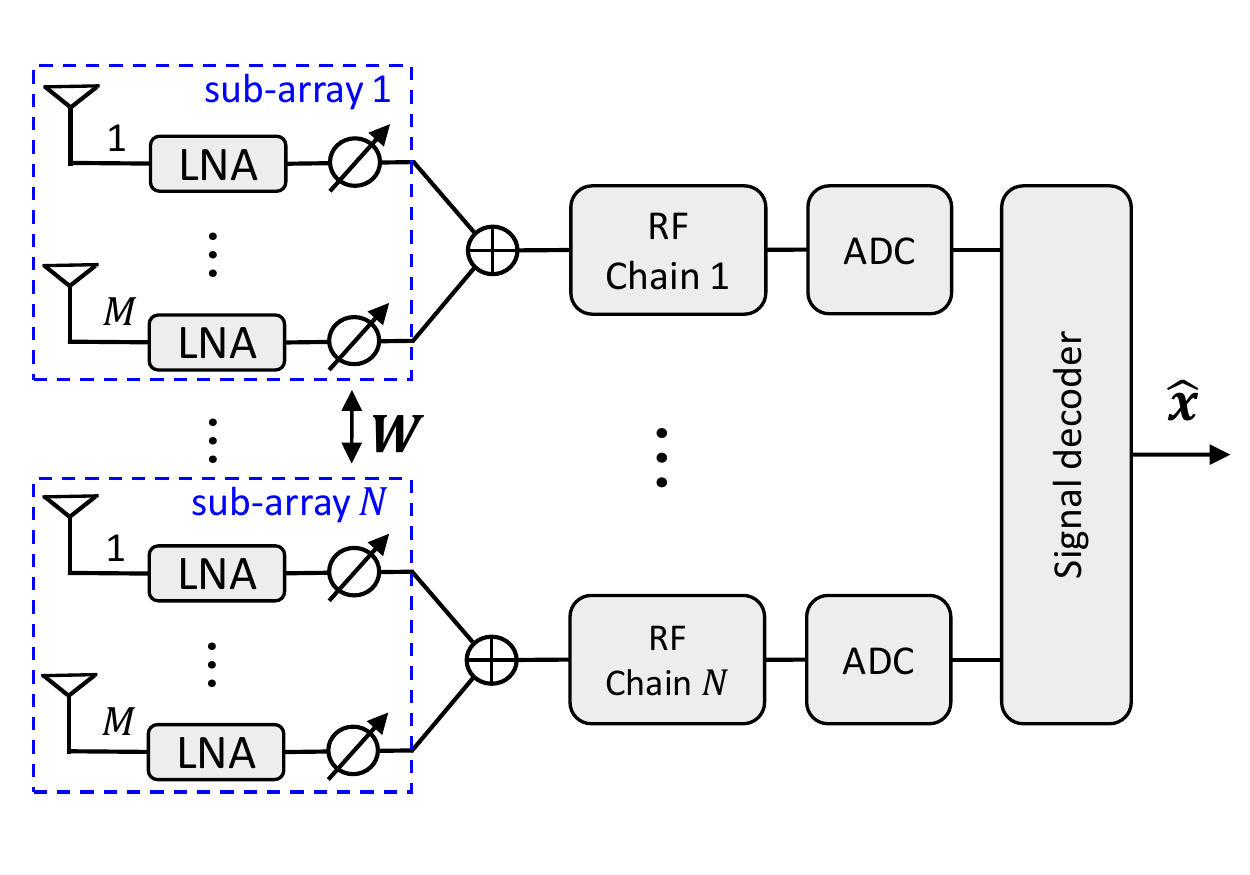}
			\label{fig:ESA_model}
		}
		\subfigure[Proposed UESA architecture]
		{
			\includegraphics[scale=0.28]{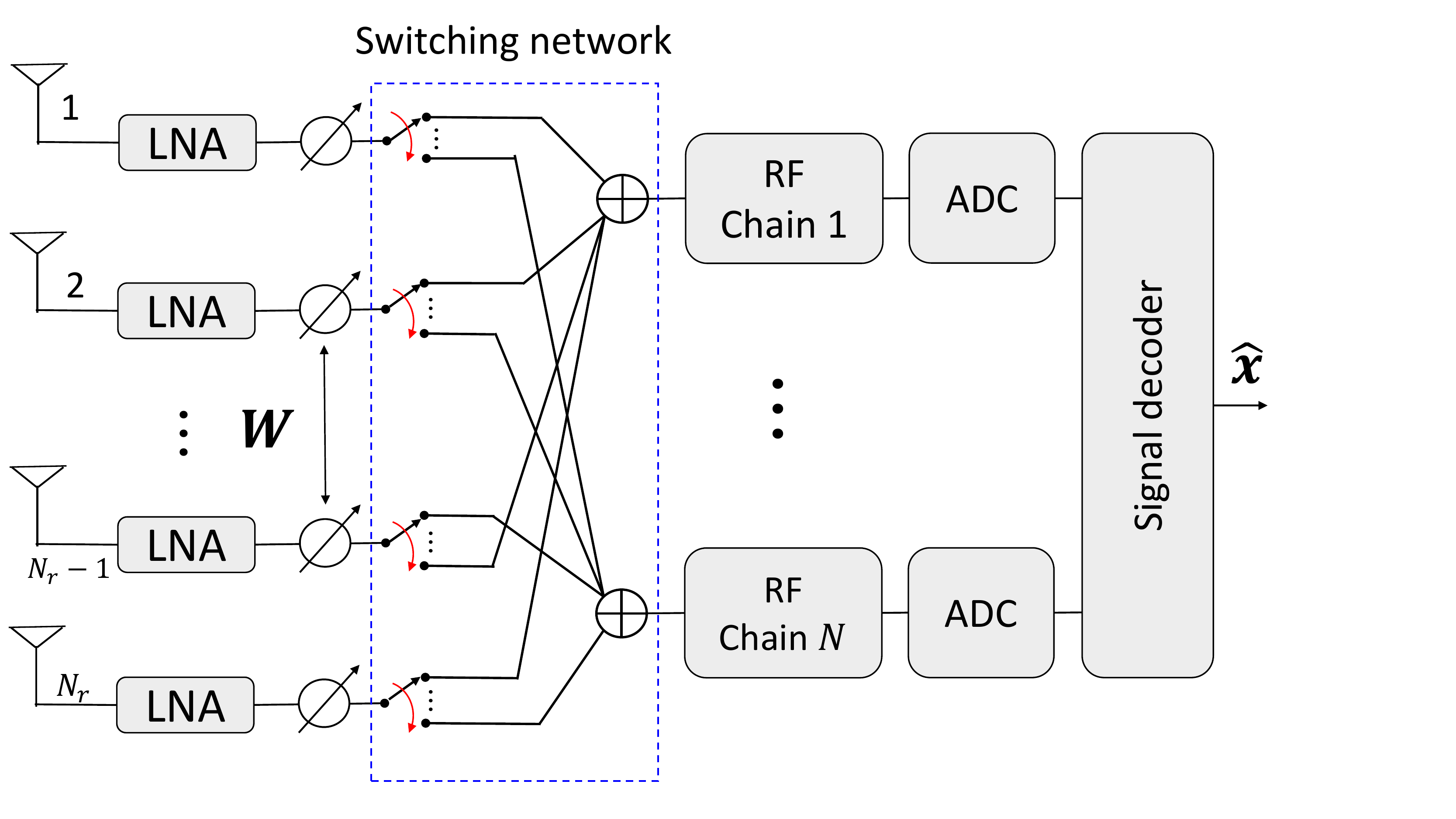}
			\label{fig:DynamicUESA}
		}
		\caption{Illustration of the conventional ESA and proposed UESA architectures.}
	\end{figure*}

%	\begin{figure}[t]
%		% Fig. 2
%		\centering
%		\includegraphics[scale=0.7]{ESA_model.pdf}
%		\caption{Conventional ESA architecture}
%		\label{fig:ESA_model}
%	\end{figure}
	
	Consider the uplink of a multi-user MIMO system consisting of a BS equipped with $N_r$ receive antennas and $N$ RF chains, where $N < N_r$, and $K$ single-antenna mobile stations (MSs). For simplicity, we assume that $N = K$. In a massive MIMO setup, $N_r$ is assumed to be much larger than $K$. In this work, for analog combining, we focus on the sub-connected architecture where each RF chain is connected to a sub-array, which is a group of phase shifters and antenna elements. The conventional sub-connected architecture is illustrated in Fig. \ref{fig:ESA_model}, where $N$ groups of $M$ antennas are equally assigned to $N$ sub-arrays, which are in turn connected to $N$ RF chains \cite{he2016energy, niu2017low}. In this work, we refer to this conventional structure as the \textit{equal sub-array} (ESA) architecture. Furthermore, we assume that following processing by the analog combiner, the received signals are passed through capacity-achieving advanced digital receivers, such as sphere decoding \cite{viterbo1999universal} and tabu search \cite{nguyen2019qr}. The analog-combined signal at the receiver can be expressed as
	\begin{align}
		\label{SM}
		\vy = \mW^H \mH \vx + \mW^H \vz,
	\end{align}
	where $\vx \in \setC^{K \times 1}$ is the vector of symbols sent from $K$ MSs. We assume that the average transmit power of each MS is normalized to one, i.e., $\mean{\abs{x_k}^2}=1, k = 1, \ldots, K$, and $\vz \in \setC^{N_r \times 1}$ is a vector of independent and identically distributed (i.i.d.) samples of additive white Gaussian noise (AWGN), where $z_i \sim \mathcal{CN}(0,N_0)$. Furthermore, $\mH  \in \setC^{N_r \times K}$ denotes the channel matrix consisting of $K$ column vectors $\vh_1, \vh_2, \ldots, \vh_K$ representing channels between $K$ MSs and the BS. Each channel entry $h_{i,k}$ represents the complex channel gain between the $k$th MS and the $i$th receive antenna of the BS. The analog combining matrix $\mW \in \setC^{N_r \times N}$ is given by
	\eq{
		\mW = 
		\begin{bmatrix}
			\vw_1  & \mathbf{0} & \mathbf{0} & \mathbf{0}\\
			\mathbf{0}  & \vw_2 & \mathbf{0} & \mathbf{0}\\
			\mathbf{0}  & \mathbf{0} & \ldots & \mathbf{0}\\
			\mathbf{0}  & \mathbf{0} & \mathbf{0} & \vw_N\\
		\end{bmatrix},
	}
	where $\vw_n = [w_{n}^{(1)}, w_{n}^{(2)}, \ldots, w_{n}^{(M)}]^T$ is the analog weighting vector for the $n$th sub-array, and $w_{n}^{(m)}$ is the $m$th element of $\vw_n$, which has the constant amplitude $1/\sqrt{M}$ but different phases, i.e., $w_{n}^{(m)} = \frac{1}{\sqrt{M}} e^{j \theta_{n}^{(m)}}, n=1,\ldots,N, m=1,\ldots,M$ \cite{gao2016energy, hur2013millimeter}.
	
	\section{Factorization-aided analog Combining and the conventional ESA system}
	
	\subsection{Factorization-aided analog combining algorithm}
	The idea of factorization-aided hybrid precoding was first proposed in \cite{gao2016energy} to optimize the hybrid precoding matrix at a transmitter employing the conventional ESA architecture. In this section, we extend it to optimize the analog combiner at the receiver. The total achievable rate $R$ for the analog-combined signal in \eqref{SM} can be expressed as \cite{el2014spatially}
	\begin{align}
		\label{rate_1}
		R = \log_2 \det \left(\mI_N +  \mR^{-1}\mW^H \mH \mH^H \mW\right),
	\end{align}
	where $\mR = N_0 \mW^H \mW$ is the noise covariance matrix after combining. Because $\mW$ has a block-diagonal structure with entries in the form of $w_{n}^{(m)} = \frac{1}{\sqrt{M}} e^{j \theta_{n}^{(m)}}$, we have $\mW^H \mW = \mI_N$. By letting $\rho = \frac{1}{N_0}$, we have $\mR^{-1} = \rho \mI_N$, and \eqref{rate_1} can be rewritten as
	\begin{align*}
		R = \log_2 \det \left(\mI_K +  \rho \mH^H \mW \mW^H \mH \right). \numberthis \label{rate_2}
	\end{align*}
	The optimal analog combiner is the solution of the problem of optimizing the total achievable rate, which can be formulated as
	\begin{align*}
		\mW^{\star} = \arg \max_{\mW} R,  \hspace{0.5cm}
		\text{s.t.} \hspace{0.5cm} \vw_1, \vw_2, \ldots, \vw_N \in \mathcal{F}, \numberthis \label{opt_1}
	\end{align*}
	where $\mathcal{F}$ is the set of feasible analog combiners.	We define $\mH_n \in \setC^{M\times K}$ as the $n$th sub-matrix of $\mH$ consisting of rows $\left\{M(n-1)+1, \ldots, Mn \right\}$ in $\mH$. Owing to the block-diagonal structure of $\mW$, we have
	\begin{align*}
		\mH^H \mW = \left[\mH_1^H \vw_1, \mH_2^H \vw_2, \ldots, \mH_N^H \vw_N \right],
	\end{align*}
	which leads to $\mH^H \mW \mW^H \mH =  \sum_{n=1}^{N} \mG_n,$ where $\mG_n = \mH_n^H \vw_n \vw_n^H \mH_n$. The following lemma expresses the total achievable rate $R$ in \eqref{rate_2} as the sum of sub-rates:
	
	\begin{lemma}
		\label{lemma:factorization_rate}
		The total achievable rate can be factorized into the sum of $N$ sub-rates corresponding to $N$ sub-arrays as follows:
		\begin{align*}
			R =  \sum_{n=1}^{N} \log_2 \left(1 + \rho \vw_n^H \mT_n \vw_n \right), \numberthis \label{rate_4}
		\end{align*}
		where $\mT_n = \mH_n \mQ_{n-1}^{-1} \mH_n^H$, with $\mQ_0 = \mI_K$ and
		\begin{align*}
			\mQ_{n-1}  = \mQ_{n-2} + \rho \mG_{n-1}. \numberthis \label{Q_n_1}
		\end{align*}
	\end{lemma}
	
	\begin{proof}
		See Appendix \ref{appendix:factorization_rate}.\epr
	\end{proof}
	
	Lemma \ref{lemma:factorization_rate} reveals that the total achievable rate $R$ can be optimized by optimizing the sub-rates 
	\begin{align}
		\label{R_n}
		R_n = \log_2 \left(1 + \rho \vw_n^H \mT_n \vw_n \right), n=1,2,\ldots,N,
	\end{align}
	where $R_n$ is the sub-rate corresponding to the $n$th sub-array. As a result, finding the optimal combining matrix $\mW^{\star}$ in \eqref{opt_1} can be factorized into finding the combining vectors $\vw_n^{\star}, n = 1, \ldots, N,$ that are the solutions of
	\begin{align*}
		\vw_n^{\star} = \arg \max_{\vw_n} R_n, \hspace{0.5cm} \text{s.t.} \hspace{0.5cm} \vw_n \in \mathcal{F}, \numberthis \label{opt_2}
	\end{align*}
	and the optimal solution is given by \cite{gao2016energy}
	\begin{align*}
		\vw_n^{\star} = \arg \min_{\vw_n \in \mathcal{F}} \norm{\vu_n^{\star} - \vw_n}^2,
	\end{align*}
	where $\vu_n^{\star}$ is the eigenvector of $\mT_n$ corresponding to its largest eigenvalue. We note that $\mT_n$ is a positive semidefinite Hermitian matrix because by \eqref{Q_n_1}, $\mQ_{n-1}$ is a positive semidefinite matrix.
		
	\begin{algorithm}[t]
		\caption{Factorization-aided analog-combining algorithm}
		\label{algorithm:factorization_AC}
		\begin{algorithmic}[1]
			%\REQUIRE $\mQ_0 = \mI$.
			%\ENSURE $\vw_1^{\star}, \vw_2^{\star}, \ldots, \vw_N^{\star}$.
			\STATE {$\mQ_{0} = \mI_K$}
			\FOR {$n = 1 \rightarrow N$} % Outer loop
			\STATE {$\mH_n = \mH(M(n-1)+1:Mn)$}
			\STATE {$\mT_n = \mH_n \mQ_{n-1}^{-1} \mH_n^H$}
			\STATE {Set $\vu_n^{\star}$ to the eigenvector corresponding to the largest eigenvalue of $\mT_n$.}
			\STATE {$\vw_n^{\star} = \mathcal{Q} (\vu_n^{\star})$}
			\STATE {$\mW^{\star}(M(n-1)+1:Mn,n) = \vw_n^{\star}$}
			\STATE {$\mG_n = \mH_n^H \vw_n^{\star} {\vw_n^{\star}}^H \mH_n$}
			\STATE {$\mQ_{n} = \mQ_{n-1} + \rho \mG_{n}$}
			\ENDFOR
		\end{algorithmic}
	\end{algorithm}

	The $n$th sub-rate $R_n$ in \eqref{R_n} can be rewritten as
	\begin{align}
		\label{R_n_Q}
		R_n = \log_2 \left(1 + \rho \vw_n^H \mH_n \mQ_{n-1}^{-1} \mH_n^H \vw_n \right).
	\end{align}
	It is evident from \eqref{Q_n_1} and \eqref{R_n_Q} that $R_n$ depends not only on $\mH_n$, but also on $\mH_{n-1}, \mH_{n-2}, \ldots, \mH_1$. Thus, given the order of $\left\{\mH_1, \mH_2, \ldots, \mH_N \right\}$, the optimization problems in \eqref{opt_2} can be solved one by one in order of $\left\{\vw_1^{\star}, \vw_2^{\star}, \ldots, \vw_{N}^{\star}\right\}$. This procedure is presented in Algorithm \ref{algorithm:factorization_AC}. In particular, in step 3, the $n$th sub-matrix $\mH_n$ is obtained from $\mH$, which allows $\mT_n$ and $\vu_n^{\star}$ to be computed in steps 4 and 5, respectively. In step 6, the combining vector $\vw_n^{\star}$ is obtained by quantizing $\vu_n^{\star}$, which ensures that the resultant analog combiner belongs to the feasible set $\mathcal{F}$. Step 8 computes $\mG_{n}$, which allows $\mQ_n$ to be updated in step 9 based on \eqref{Q_n_1}.

	\subsection{Conventional ESA architecture}

	Let $\sigma_{1}(n)$ be the largest eigenvalue of $\mT_n$. The total achievable rate in \eqref{rate_4} is upper-bounded by
	\begin{align*}
	R_{\text{ESA}} \leq R_{\text{ESA}}^{\text{ub},1} = \sum_{n=1}^{N} \log_2 \left[1 + \rho \sigma_{1}(n)\right]. \numberthis \label{rate_5}
	\end{align*}
	The equality occurs when $\vw_n^{\star} = \vu_n^{\star}, n=1,2,\ldots,N$, i.e., quantization is not applied to generate $\vw_n^{\star}$. By applying Jensen's inequality, we have
	\begin{align*}
		R_{\text{ESA}}^{\text{ub},1} \leq  R_{\text{ESA}}^{\text{ub}} = N \log_2 \left(1 + \frac{\rho}{N} \sum_{n=1}^{N}  \sm{n} \right), \numberthis \label{rate_8}
	\end{align*}
	where the equality occurs when $\sm{1} = \sm{2} = \ldots = \sm{N}$, i.e., the largest eigenvalues of $\mT_1, \mT_2, \ldots, \mT_N $ are the same.
	
	\begin{lemma}
		\label{lemma:trace_Tn_decrease}
		In the conventional ESA system, ${\tr{\mT_n}}$ is a descending function of $n$ when $N_r$ grows while $K$ is kept constant.
	\end{lemma}
	
	\begin{proof}
		See Appendix \ref{appendix:prove_traceTn_decrease}. \epr
	\end{proof}
	
	\begin{remark}
		By Lemma 2, ${\tr{\mT_n}}$, which is the sum of all eigenvalues of $\mT_n$, decreases with $n$ in the ESA system. Under the typical assumption for massive MIMO systems whereby $N_r$ is sufficiently larger than $K$, it is nearly impossible for the largest eigenvalues $\sm{n}$, $n=1,2,\ldots,N$ to satisfy the condition on equality in \eqref{rate_8}, i.e., $\sm{1} = \ldots = \sm{N}$ in the ESA system. This is also confirmed by the simulation results in Section V, which shows that similar to ${\tr{\mT_n}}$, $\sm{n}$ tends to decrease with $n$.
		
		This result shows that in the ESA system, it is difficult to reach the upper bound of the total achievable rate in \eqref{rate_8}, which motivates us to design a new architecture called the UESA in the next section.
	\end{remark} 
	
	\section{Proposed UESA architecture}
	
	To improve performance, we propose assigning unequal numbers of antennas to sub-antenna arrays. In the UESA architecture, the receive antennas are dynamically connected to RF chains via a switching network, as illustrated in Fig. \ref{fig:DynamicUESA}. The switching network allows the UESA scheme to adaptively assign antennas to sub-antenna arrays or, equivalently, to adaptively connect antennas to RF chains. The feasibility of switch-based analog beamforming has been widely considered in the literature. For switching, tunable switches \cite{mendez2016hybrid, bogale2016number, alkhateeb2016massive} or digital chips \cite{zhu2016adaptive} with low power consumption, low hardware costs, and high turning speed \cite{schmid2014analysis} can be used.

	Let $m_n$ be the number of antennas assigned to the $n$th sub-antenna array, which is in the $n$th iteration of Algorithm 1. We then have $m_n \geq 1$ and $\sum_{n=1}^{N} m_n = N_r$. Therefore, in the UESA architecture, $m_n$ is a value satisfying $1 \leq m_n \leq N_r-N+1$. For simplicity, we define $\comM$ as the set of numbers of antennas assigned to sub-arrays one to $N$, i.e., $\comM = \left\{m_1, m_2, \ldots, m_N\right\}$.
	
	\subsection{Design of the UESA architecture}
	In the following analysis, to distinguish the proposed UESA system from the conventional ESA system, the largest eigenvalue of $\mT_n$ in the UESA system is denoted by $\mu_1(n)$. In a similar manner to \eqref{rate_5} and \eqref{rate_8}, we obtain the upper bounds of the total achievable rate of the proposed UESA architecture $R_{\text{UESA}}$ as follows:
	\begin{align*}
		R_{\text{UESA}} \leq R_{\text{UESA}}^{\text{ub},1} = \sum_{n=1}^{N} \log_2 \left(1 + \rho \um{n}\right), \numberthis \label{ub_uesa_1}
	\end{align*}
	and
	\begin{align*}
		R_{\text{UESA}}^{\text{ub},1} 
		\leq  R_{\text{UESA}}^{\text{ub}}
		= N \log_2 \left(1 + \frac{\rho}{N} \sum_{n=1}^{N}  \um{n} \right). \numberthis \label{ub_uesa}
	\end{align*}
	The equality in \eqref{ub_uesa_1} is obtained when $\vw_n^{\star} = \vu_n^{\star}, n=1,2,\ldots, N$ and that in \eqref{ub_uesa} is obtained when $\um{1} = \um{2} = \ldots = \um{N}$. To optimize the sum rate, based on \eqref{ub_uesa}, we design the UESA architecture such that $R_{\text{UESA}}^{\text{ub},1} \approx R_{\text{UESA}}^{\text{ub}}$, i.e., 
	\begin{align*}	
		\um{1} \approx \um{2} \approx \ldots \approx \um{N}. \numberthis \label{mu_equal}
	\end{align*}
	At the same time, $R_{\text{UESA}}^{\text{ub}}$ should be maximized, and this can be achieved by maximizing $\sum_{n=1}^{N}  \um{n}$. To achieve these objectives, we need to optimize $\um{n}$, $n=1,2,\ldots,N$. However, it is difficult to directly optimize them, because of which we consider ${\tr{\mT_n}}$ for optimization.
	
	In particular, considering that $\um{n}$ is the largest eigenvalue of $\mT_n$, we relax \eqref{mu_equal} to 
	\begin{align*}
		\tr{\mT_1} \approx \tr{\mT_2} \approx \ldots \approx {\tr{\mT_N}}. \numberthis \label{trace_equal}
	\end{align*}
	The simulation results in Section \ref{sec:simulation_result} numerically verify that by achieving \eqref{trace_equal}, $\um{n}$, $n=1,2,\ldots,N$, can approach condition \eqref{mu_equal}.
	
	Similar to \eqref{trace_equal}, because it is difficult to directly maximize $\sum_{n=1}^{N}  \um{n}$, we instead maximize its upper bound $\sum_{n=1}^{N} {\tr{\mT_n}}$. Consequently, we consider the following design objectives for the UESA:
	\begin{itemize}
		\item $\mathcal{O}_1: {\tr{\mT_1}} \approx \ldots \approx {\tr{\mT_N}}$
		\item $\mathcal{O}_2: \sum_{n=1}^{N} {\tr{\mT_n}}$ is maximized
	\end{itemize}
	\vspace{0.1cm}
	\subsubsection{UESA design for $\mathcal{O}_1$}
%	\begin{figure}[t]
%		% Fig. 2
%		\centering
%		\includegraphics[scale=0.28]{DynamicUESA.pdf}
%		\caption{Illustration of the UESA architecture}
%		\label{fig:DynamicUESA}
%	\end{figure}
	As discussed in Remark 1, it is unlikely that objective $\mathcal{O}_1$ is satisfied in the conventional ESA system because $M$ is a constant value in \eqref{tr_3}. Therefore, in the UESA system, we propose assigning different numbers of antennas to sub-antenna arrays. In other words, $m_1, m_2, \ldots, m_N$ are not necessarily the same, and for a given set $\comM = \left\{m_1, m_2, \ldots, m_N\right\}$, the following theorem suggests a constraint on $\comM$ for optimizing the achievable rate.

	\begin{theorem}
		\label{theorem:m_increase}
		To satisfy the objective $\mathcal{O}_1$, $\comM$ should be arranged in non-decreasing order, i.e., $m_1 \leq m_2 \leq \ldots \leq m_N$.
	\end{theorem}
	
	\begin{proof}
		See Appendix \ref{appendix:proof_of_theorem1} \epr
	\end{proof}
	
	Antenna allocation based on Theorem 1 can yield objective $\mathcal{O}_1$, which makes $R_{\text{UESA}}$ closer to its upper bound $R_{\text{UESA}}^{\text{ub}}$ and yields the enhanced achievable rate of the UESA system. This also implies that $R_{\text{UESA}}^{\text{ub}}$ becomes a tighter upper bound for the total achievable rate compared with $R_{\text{ESA}}^{\text{ub}}$. For further performance improvement, under constraint $m_1 \leq m_2 \leq \ldots \leq m_N$, we consider the second design objective $\mathcal{O}_2$.
	
	\vspace{0.1cm}
	\subsubsection{UESA design for $\mathcal{O}_2$}
	
	To achieve $\mathcal{O}_2$, we consider the following theorem:
	
	\begin{theorem}
		\label{theorem:channel_ordering}
		When $N_r$ grows while $K$ is kept constant, $\sum_{n=1}^{N} \tr{\mT_n}$ can be maximized by sorting the rows of $\mH$ in decreasing order of their norms.
	\end{theorem}
	
	\begin{proof}
		In a UESA architecture, $\mH_n$ has size $m_n \times K$. If $N_r$ grows while $K$ is fixed in a massive MIMO system, $m_n$ becomes much larger than $K$. Then, we have \cite{lu2014overview}
		\begin{align*}
			{\mH_n^H \mH_n} \approx \text{diag} \left\{{\norm{\vh_{1}^{(n)}}}^2, {\norm{\vh_{2}^{(n)}}}^2, \ldots, {\norm{\vh_{K}^{(n)}}}^2 \right\},
		\end{align*}
		which is a diagonal matrix of size $K \times K$ with ${\norm{\vh_{1}^{(n)}}}^2, {\norm{\vh_{2}^{(n)}}}^2, \ldots, {\norm{\vh_{K}^{(n)}}}^2 $ on the main diagonal, where $\vh_{k}^{(n)}$ is the $k$th column of $\mH_n$. Therefore, \eqref{tr_2} can be rewritten as
		\begin{align*}
			{\tr{\mT_n}} 
			&= \sum_{k=1}^{K} {{\norm{\vh_{k}^{(n)}}}^2} {q_{k,k}^{(n-1)}}\\ 
			&\approx {{\norm{\vh^{(n)}}}^2} \sum_{k=1}^{K} {q_{k,k}^{(n-1)}} \numberthis \label{mean_tr_11} \\
			&= {{\norm{\vh^{(n)}}}^2} {\tr{\mQ_{n-1}^{-1}}} ,  \numberthis \label{mean_tr_2}
		\end{align*}
		where $q_{k,k}^{(n-1)}$ is the $k$th diagonal element of $\mQ_{n-1}^{-1}$. Because ${\norm{\vh_{k}^{(n)}}}^2$ is approximately the same for all $k$, we can write ${{\norm{\vh_{k}^{(n)}}}^2} \approx {\norm{\vh^{(n)}}}^2, k = 1,2,\ldots,K$, and with the note that ${\tr{\mQ_{n-1}^{-1}}} = \sum_{k=1}^{K} {q_{k,k}^{(n)}}$, we obtain \eqref{mean_tr_11} and \eqref{mean_tr_2}. 
		Letting
		\begin{align*}
			\boldsymbol{\phi}_{\vh} 
			&= \left[{{\norm{\vh^{(1)}}}^2}, \ldots, {{\norm{\vh^{(N)}}}^2}\right]^T,\\
			\boldsymbol{\phi}_{\mQ} 
			&=  \left[\tr{\mQ_{0}^{-1}}, \ldots, \tr{\mQ_{N-1}^{-1}}\right]^T,
		\end{align*}
		we obtain
		\begin{align*}
			\sum_{n=1}^{N} {\tr{\mT_n}} = \boldsymbol{\phi}_{\vh}^T \boldsymbol{\phi}_{\mQ}. \numberthis \label{sum_mean_tr}
		\end{align*}
		We note that for a given $\boldsymbol{\phi}_{\mQ} $, \eqref{sum_mean_tr} is maximized if $\boldsymbol{\phi}_{\vh}$ and $\boldsymbol{\phi}_{\mQ} $ are in parallel as much as possible. Because ${\tr{\mQ_{n-1}^{-1}}}$ decreases with $n$ as proved in Appendix \ref{appendix:prove_traceTn_decrease}, the elements of $\boldsymbol{\phi}_{\mQ} $ are in decreasing order. Furthermore, in Appendix \ref{appendix:decreasing_rate_with_channel_ordering}, where the decreasing rate of ${\tr{\mQ_{n-1}^{-1}}}$ is considered, we show that it does not significantly change with channel ordering. Therefore, from \eqref{sum_mean_tr}, it can be concluded that $\boldsymbol{\phi}_{\vh}^T \boldsymbol{\phi}_{\mQ}$ can be maximized if the rows of $\mH$ are ordered such that ${{\norm{\vh_{k}^{(n)}}}^2}$ also decreases with $n$. Under the constraint $m_1 \leq m_2 \leq \ldots \leq m_N$, the channel rows are ordered in  decreasing order of their norms. \epr
	\end{proof}
	
	Theorems \ref{theorem:m_increase} and \ref{theorem:channel_ordering} give insights into the design of the UESA architecture that allow for improvements in the total achievable rate of the conventional ESA architecture. However, they do not guarantee the optimal total rate because a large number of combinations $\comM$ satisfies $m_1 \leq m_2 \leq \ldots \leq m_N$. In the subsections below, we propose algorithms to find the optimal number of antennas for each sub-antenna array. Theorems 1 and 2 are used to find the near-optimal numbers of antennas in sub-arrays with low computational complexity.
	
	\subsection{Algorithms for antenna allocation in UESA}
	\subsubsection{UESA with exhaustive search (UESA-ES)}
	Let $\comMs = \left\{m_1^{\star}, m_2^{\star}, \ldots, m_N^{\star}\right\}$ be the set of optimal numbers of antennas for sub-arrays. This provides the highest achievable rate among all possible candidates of $\comM$. Channel ordering based on Theorem \ref{theorem:channel_ordering} is first performed, and the ES algorithm then examines all candidates in
	\begin{align*}
		\mathbb{S} = \left\{\comM: 1 \leq m_n \leq N_r-N+1, \sum_{n=1}^{N} m_n = N_r\right\}
	\end{align*}
	to determine $\comMs$.
	For each candidate, Algorithm \ref{algorithm:factorization_AC} is used to find the analog combining matrix $\mW^{\star}$. Finally, $\comMs$ is set to the candidate that provides the highest sum rate. Then, $\abs{\mathbb{S}}$ becomes very large, especially for massive MIMO systems, and this results in excessively high computational complexity. To resolve this problem, we propose an algorithm that performs search in a subset of $\mathbb{S}$ with much lower complexity.
	
	\subsubsection{UESA with reduced-ES (UESA-RES)}
	
	\begin{algorithm}
		\caption{Analog combining with the UESA-RES algorithm}
		\label{algorithm:RES}
		\begin{algorithmic}[1]
			\REQUIRE $\mathbb{S}_r$
			\ENSURE $\mW_{\text{UESA-RES}}^{\star},\comM^{\star}$
			
			\STATE{Obtain $\underline{\mH}$ by ordering the channel rows in decreasing order of their norms.}
			\STATE {$\tau = 0$}
			\FOR {$i = 1 \rightarrow \abs{\mathbb{S}_r}$} % Outer loop
			\STATE {Set $\comM$ to the $i$th candidate in $\mathbb{S}_r$.}
			\STATE {Use Algorithm \ref{algorithm:factorization_AC} to find $\mW_{\text{UESA}}$ for $\comM$ and $\underline{\mH}$.}
			\STATE {$R_{\text{UESA}} = \log_2 \det \left(\mI_K +  \rho \underline{\mH}^H \mW_{\text{UESA}} \mW^H_{\text{UESA}} \underline{\mH} \right)$}
			\IF {$R_{\text{UESA}} > \tau$}
			\STATE {$\mW_{\text{UESA-RES}}^{\star} = \mW_{\text{UESA}}$}
			\STATE {$\comM^{\star} =\comM$}
			\STATE {$\tau = R_{\text{UESA}}$}
			\ENDIF
			\ENDFOR
		\end{algorithmic}
	\end{algorithm}
	
	Based on Theorem \ref{theorem:m_increase}, instead of searching for $\comMs$ in the entire space $\mathbb{S}$, the RES algorithm performs search only in $\mathbb{S}_r$,  a sub-set of $\mathbb{S}$, which features $\comM$ such that $m_1 \leq m_2 \leq \ldots \leq m_N$, i.e.,
	\begin{align*}
		\mathbb{S}_r = \left\{\comM: \comM \in \mathbb{S}, m_1 \leq m_2 \leq \ldots \leq m_N \right\}.
	\end{align*}

	The UESA-RES algorithm is summarized in Algorithm \ref{algorithm:RES}. After ordering the channel rows based on Theorem \ref{theorem:channel_ordering} in step 1, the candidates in $\mathbb{S}_r$ are sequentially tested. Each candidate $\comM$ is used as input to Algorithm \ref{algorithm:factorization_AC} to find the combining matrix and the corresponding total achievable rate as shown in steps 5 and 6, respectively. In steps 7--11, $\mW_{\text{UESA-RES}}^{\star}$ and $\comM^{\star}$ are updated whenever a combining matrix $\mW_{\text{UESA}}$ and $\comM$ with a higher sum rate are found. Simulation results in Section \ref{sec:simulation_result} show that the RES algorithm provides almost identical performance to the ES algorithm while incurring considerably lower computational complexity.
	\vspace{0.2cm}
	\subsubsection{UESA with RES and early termination (UESA-RES-ET)}
	
	In the search for $\comMs$, we first examine $\comM$ with small differences among its elements. Alternatively, we can first examine $\comM$ with large differences among its elements. The optimal search order depends on how fast $\tr{\mQ_{n-1}^{-1} }$ decreases with $n$ in \eqref{m_trace}. Through simulations, we found that $\tr{\mQ_{n-1}^{-1} }$ tends to decrease substantially quickly with $n$, which implies that $\comM$ with large differences between neighboring elements is more likely to be $\comMs$. Based on this observation, in the UESA-RES-ET algorithm, candidates for $\comM$ in $\mathbb{S}_r$ are sorted according to the differences between neighboring elements of $\comM$. Furthermore, if no better candidate is found over a certain number of iterations, the search process is terminated.
	
	The UESA-RES-ET algorithm is summarized in Algorithm \ref{algorithm:UESA_RES_EE}. A vector $\mathbf{\Pi} = \left\{ \pi_1, \ldots, \pi_{\abs{\mathbb{S}_r}} \right\}$ corresponding to $\abs{\mathbb{S}_r}$ candidates in $\mathbb{S}_r$ is first computed. Specifically, for a candidate $\comM = \left\{m_1, \ldots, m_n, m_{n+1}, \ldots, m_N\right\}$, $\pi_i$ is given by $\pi_i = \prod_{n=1}^{N-1} \left(d_n + 1\right)$, as shown in step 5, where $d_n = \abs{m_{n+1} - m_n}, d_n \in \left[0, N_r-N\right]$ is the difference between $m_{n+1}$ and $m_n$, which is computed in step 4\footnote{We employ this form of $\pi_i$ instead of the sum of differences, i.e., $\sum_{n=1}^{N-1} d_n$, because the latter only measures the difference between $m_1$ and $m_N$, i.e., $ \sum_{n=1}^{N-1} d_n = \abs{d_{1} - d_{N-1}}$, and a larger $\sum_{n=1}^{N-1} d_n$ does not imply larger differences among the elements of $\comM$.}. A larger $\pi_i$ does, however, imply larger differences among the elements of $\comM$. 
	As a result, $\comM$ with large $\pi_i$ is considered as a more promising candidate, and should be tested first. Therefore, in step 7, $\underline{\mathbb{S}}_r$ is obtained by sorting the elements of $\mathbb{S}_r$  in descending order of elements of $\mathbf{\Pi}$, and the candidates in $\underline{\mathbb{S}}_r$ are tested one by one to determine $\mW_{\text{UESA-RES-ET}}^{\star}$ and  $\comMs$ in steps 10--26. For early termination, a counter $count$ is used to record the number of iterations for which new candidates do not provide any further performance improvement. For each candidate $\comM$  that does not improve the achievable rate, $count$ increases by one, and is otherwise reset to zero, as shown in steps 16 and 21, respectively. When $count_{max}$ is reached, it is likely that the optimal candidate has already been searched, and thus the search process is  terminated in step 24.
	
	\begin{algorithm}
		\caption{Analog combining with the UESA-RES-ET algorithm}
		\label{algorithm:UESA_RES_EE}
		\begin{algorithmic}[1]
			\REQUIRE $\mathbb{S}_r, count_{max}$
			\ENSURE $\mW_{\text{UESA-RES-ET}}^{\star},\comM^{\star}$
			
			%\FOR {$m_1 = 1 \rightarrow N_r-N+1$} % 1st Outer loop
			
			\STATE {$\mathbf{\Pi} = [0, 0, \ldots, 0],  \left(\abs{\mathbb{S}_r} \text{zeros}\right)$}
			\FOR {$i = 1 \rightarrow \abs{\mathbb{S}_r}$}
			\STATE {Set $\comM = \left\{m_1, m_2, \ldots, m_N \right\}$ to the $i$th candidate in $\mathbb{S}_r$.}
			
			\STATE {Compute $\left\{d_1, d_2, \ldots, d_{N-1}\right\}$ with $d_n = \abs{m_{n+1} - m_n}$.}
			
			\STATE {Compute $\pi_i = \prod_{n=1}^{N-1} \left(d_n + 1\right)$ and set it to the $i$th element of $\mathbf{\Pi}$.}
			\ENDFOR
			\STATE {Obtain $\underline{\mathbb{S}}_r$ by sorting the candidates in $\mathbb{S}_r$ in descending order of the elements of $\mathbf{\Pi}$.}
			
			\STATE{Obtain $\underline{\mH}$ by ordering the channel rows in descending order of their norms.}
			\STATE {$\tau = 0, count = 0$}
			\FOR {$i = 1 \rightarrow \abs{\underline{\mathbb{S}}_r}$} % Outer loop
			\IF {$count < count_{max}$}
			\STATE {Set $\comM$ to the $i$th candidate in $\underline{\mathbb{S}}_r$.}
			\STATE {Use Algorithm \ref{algorithm:factorization_AC} to find $\mW_{\text{UESA}}$ for $\comM$ and $\underline{\mH}$.}
			\STATE {$R_{\text{UESA}} = \log_2 \det \left(\mI_K +  \rho \underline{\mH}^H \mW_{\text{UESA}} \mW^H_{\text{UESA}} \underline{\mH} \right)$}
			\IF {$R_{\text{UESA}} < \tau$}
			\STATE {$count = count + 1$}
			\ELSE
			\STATE {$\mW_{\text{UESA-RES-ET}}^{\star} = \mW_{\text{UESA}}$}
			\STATE {$\comM^{\star} =\comM$}
			\STATE {$\tau = R_{\text{UESA}}$}
			\STATE {$count = 0$}
			\ENDIF
			\ELSE
			\STATE {Terminate the search process.}
			\ENDIF
			\ENDFOR
			%\ENDFOR
		\end{algorithmic}
	\end{algorithm}
	
	% new ====================================================
	\subsubsection{UESA with fast antenna allocation (Fast-UESA)}
	Although UESA-RES and UESA-RES-ET significantly reduce the complexity of antenna allocation, they still need to examine a large number of candidates in the case of large $\abs{\underline{\mathbb{S}}_r}$. In order to further reduce complexity, we propose a fast antenna allocation algorithm for UESA (Fast-UESA).
	
	Let $\comM = \left\{m_1, \ldots, m_N \right\}$ be the first candidate in $\underline{\mathbb{S}}_r$. The Fast-UESA algorithm starts by taking $\comM$ as the initial solution and subsequently updating it based on objectives $\mathcal{O}_1$ and $\mathcal{O}_2$ in order to achieve improved sum rate. Specifically, in order to increase the upper bound of the sum rate $(R_{\text{ESA}}^{\text{ub}})$ and to make it reachable, we update the elements of $\comM$ such that the largest eigenvalues $\mu_1(1), \ldots, \mu_1(N)$ are enhanced while minimizing their differences. Intuitively, if $\mu_1(n)$ is much smaller than the others, more antennas should be assigned to the $n$th sub-array, i.e., $m_n$ should be increased. However, an increase in $m_n$ can make the constraint $\sum_n m_n = N_r$ unsatisfied. To solve this problem, we set the last element of $\comM$ to $m_N = N_r - \sum_{n=1}^{N-1} m_n$. The motivation for this approach is that in $\comM$, the last element $m_N$ is typically much larger than the others. We recall that $\comM$ is the first candidate in $\underline{\mathbb{S}}_r$ and that it has maximal distances between elements, as described in the UESA-RES-ET algorithm. For example, for $N_r=64, N=4$, we have $\comM = \left\{1,7,17,39\right\}$. An increase of $m_1$ by one leads to a decrease of $m_N=39$ by one, which can significantly improve the sum rate of the first sub-array but does not significantly affect the sum rate of the last sub-array.
	
	The Algorithm \ref{algorithm:fast-uesa} summarizes the Fast-UESA algorithm with the input of the first element in $\underline{\mathbb{S}}_r$ and the ordered channel matrix. Steps 1--3 initialize the algorithm by finding $\mW_{\text{UESA}}$ and $\mu_1(1), \ldots, \mu_1(N)$, computing the initial sum rate $\tau$, and assigning $\comM$ to $\comM^{\star}$. Elements of $\comM$ are then updated over $I$ iterations, as presented in steps 5--26. First, in step 6, $\Delta_n, n=1, \ldots, N$, are computed. These metrics measure the differences between $\mu_1(n), n=1, \ldots, N$, and the average value of $\left\{\mu_1(1), \ldots, \mu_1(N)\right\}$; these differences affect objective $\mathcal{O}_1$. For example, if $\Delta_n \ll 0$ and $\Delta_k \approx 0, k \neq n$, $\mu_1(n)$ becomes much smaller than $\mu_1(k), k \neq n$. In that case, it will be difficult to guarantee objectives $\mathcal{O}_1$ and $\mathcal{O}_2$. Therefore, we compare $\Delta_n$ to a predefined threshold $\gamma$, which is set in step 4, to decide whether more antennas should be assigned to the $n$th sub-array, as shown in steps 8--10. In step 11, $m_N$ is adjusted to guarantee that the constraint $\sum_{n=1}^{N-1} m_n = N_r$ is satisfied. In steps 14--21, only those $\comM$ that are member of $\underline{\mathbb{S}}_r$ are examined. Specifically, in step 15, we obtain $\mu_1(1), \ldots, \mu_1(N)$ based on Algorithm 1, which are then used in step 16 as condition for updating $\comM^{\star}$. This condition guarantees that the updated $\comM^{\star}$ will achieve larger $\sum_{n=1}^{N} \mu_1(n)$, which affects the upper bound of the total achievable rate, as discussed in Section IV-A. We note here that this step is different from the corresponding steps in the UESA-RES and UESA-RES-ET algorithms, which check the total achievable rate for updating $\comM^{\star}$. This contributes to a reduction in the complexity of the Fast-UESA algorithm. In step 23, the update process is terminated early in the case that $m_N < m_{N-1}$. Finally, $\mW_{\text{Fast-UESA}}^{\star}$ and $\comM^{\star}$ become the best ones found so far.
	
	\begin{algorithm}
		\caption{Analog combining with Fast-UESA}
		\label{algorithm:fast-uesa}
		\begin{algorithmic}[1]
			\REQUIRE The first candidate $\comM = \left[m_1, m_2, \ldots, m_N \right]$ of $\underline{\mathbb{S}}_r$, and the ordered channel matrix $\underline{\mH}$
			\ENSURE $\mW_{\text{Fast-UESA}}^{\star},\comM^{\star}$
			\STATE {Use Algorithm \ref{algorithm:factorization_AC} to find ${\mu_1(n)}, n=1,\ldots,N$ and $\mW_{\text{UESA}}$ for $\comM$ and $\underline{\mH}$.}
			
			\STATE {$\tau = \sum_{n=1}^{N} \mu_1(n)$}
			\STATE {$\comM^{\star} = \comM$}

			\STATE {Set a threshold $\gamma$.}
			
			\FOR {$i = 1 \rightarrow I$}
				
				\STATE {$\Delta_n = \mu_1(n) - \frac{1}{N} \sum_{n} \mu_1(n), n=1,\ldots,N$}
				
				\FOR {$n = 1 \rightarrow N$}
				
					\IF {$\Delta_n < \gamma$}
						\STATE {$m_n = m_n + 1$}
					\ENDIF 
					\STATE {$m_N = N_r - \sum_{n=1}^{N-1} m_n$}
					
					\IF {$m_N > m_{N-1}$}
						\STATE {Update $\comM = \left[m_1, m_2, \ldots, m_N \right]$.}
						\IF {$\comM \in \underline{\mathbb{S}}_r$}
							\STATE {Use Algorithm \ref{algorithm:factorization_AC} to find ${\mu_1(n)}, n=1,\ldots,N$ and $\mW_{\text{UESA}}$ for $\comM$ and $\underline{\mH}$.}

							\IF {$\sum_{n=1}^{N} \mu_1(n) > \tau$}
								\STATE {$\comM^{\star} =\comM$}
								\STATE {$\mW_{\text{Fast-UESA}}^{\star} = {\mW}_{\text{UESA}}$}
								\STATE {$\tau = R_{\text{UESA}}$}
								
							\ENDIF 
						\ENDIF
					\ELSE 
					\STATE {Break to terminate the updating process.}
				\ENDIF
			\ENDFOR

		\ENDFOR
			%\ENDFOR
		\end{algorithmic}
	\end{algorithm}

	% ========================================================
	\subsection{Power consumption and hardware cost}
	
	In the ESA and UESA architectures illustrated in Figs. \ref{fig:ESA_model} and \ref{fig:DynamicUESA}, respectively, each receive antenna needs a low-noise amplifier (LNA) and a phase shifter while each RF chains requires an ADC. Let $P_{\text{LNA}}, P_{\text{PS}}, P_{\text{SW}},  P_{\text{RF}}$, and $P_{\text{ADC}}$ be the power consumed by the LNA, phase shifter, switch, RF chain, and ADC, respectively. Then, the total power consumed by the ESA and UESA architectures can be expressed as 
	\begin{align*}
		&P_{\text {\text{ESA}}} = N_r \left(P_{\text{LNA}}+P_{\text{PS}}\right) + N \left(P_{\text{RF}} + P_{\text{ADC}}\right) \numberthis \label{P_ESA}
	\end{align*}
	and
	\begin{align*}
		P_{\text {\text{UESA}}} &= N_r \left(P_{\text{LNA}}+P_{\text{PS}}\right) + N \left(P_{\text{RF}} + P_{\text{ADC}}\right) + N_r P_{\text{SW}}, \numberthis \label{P_UESA}
	\end{align*}
	respectively. It is observed from \eqref{P_ESA} and \eqref{P_UESA} that the UESA system additionally requires power consumption  of $N_r P_{\text{SW}}$ compared with the conventional ESA system due to the use of a switching network. However, because the power consumed by switches is small  \cite{mendez2016hybrid, celik2006implementation, alkhateeb2016massive, mendez2015channel, schmid2014analysis}, the corresponding increases in power consumption by the UESA scheme are not significant. Specifically, \cite{mendez2016hybrid} and \cite{alkhateeb2016massive} demonstrate that the power consumption of a switch is six times lower than that of a phase shifter and 40 times lower than that of an ADC block. Therefore, the power consumption increases due to the switching network do not significantly affect the total power consumption of the sub-connected architecture. We further investigate this issue in the next section.
	
	The hardware cost of a receiver significantly depends on the number of RF chains. The proposed UESA architecture does not require any additional RF chains compared with the conventional ESA architecture. Furthermore, switches are low-cost devices\cite{mendez2016hybrid, celik2006implementation}, and this fact results in the relatively low cost of the switching network in the proposed UESA architecture. Therefore, the proposed UESA architecture can still be a cost-effective architecture for hybrid beamforming in massive MIMO systems.
	
	\section{Simulation results}
	\label{sec:simulation_result}
	In this section, we numerically evaluate the achievable rate, power consumption, energy efficiency, and computational complexities of the proposed UESA architecture. In simulations, the channel coefficients between each MS and the BS are generated based on the geometric Saleh--Valenzuela channel model,  which is a typical channel model for millimeter-wave communication systems \cite{han2015large, lee2015hybrid, chen2015iterative, gao2016energy}. Specifically, the channel vector between the BS and the $k$th MS can be expressed as \cite{bogale2016number, alkhateeb2015limited, choi2016analog}
	\begin{align}
		\label{channel_model}
		\vh_k = \sqrt{\frac{N_r}{L_k}} \sum_{l=1}^{L_k} \alpha_{k,l}  \va_{BS} (\phi_{k,l}), 
	\end{align}
	where $L_k$ is the number of effective channel paths corresponding to a limited number of scatters between the BS and the $k$th MS, $\alpha_{k,l}$ and $\phi_{k,l}$ are the gain and the azimuth angle of arrival (AoA) of the $l$th path, respectively. All channel path gains $\alpha_{k,l}$  are assumed to be i.i.d. Gaussian random variables with zero mean and unit variance. Furthermore, $\va_{BS}$ represents the normalized receive array response vector at the BS that depends on the structure of the antenna array. In this work, we considered a uniform linear array (ULA) where the array response vector is given by \cite{gao2016energy, alkhateeb2015limited}
	\begin{align*}
		\va (\phi) = \frac{1}{\sqrt{N_r}} [1, e^{j \frac{2\pi}{\lambda} d \sin(\phi)}, \ldots, e^{j (N_r-1) \frac{2\pi}{\lambda} d \sin(\phi)}]^T, \numberthis \label{array_response}
	\end{align*}
	where $\lambda$ denotes the wavelength of the signal and $d$ is antenna spacing in the antenna array. 
	
	For simplicity, we assume an identical number of effective channel paths between each MS and the BS, which is set to $L_k = 10$ for $k = 1, 2, \ldots, K$ \cite{nguyen2016hybrid, park2017dynamic, el2013multimode, el2014spatially}. The AoAs $\phi_{k,l}$ are assumed to be uniformly distributed in $\left[0; 2\pi\right]$ \cite{nguyen2016hybrid, li2017hybrid}. The ULA model is employed for the receive antenna array at the BS with antenna spacing of half a wavelength, i.e., $\frac{d}{\lambda} = \frac{1}{2}$ in \eqref{array_response} \cite{gao2016energy, bogale2016number}. The phases in the analog combiner are restricted to $ \Theta = \left\{ 0, \frac{2\pi}{Q}, \frac{4\pi}{Q}, \ldots, \frac{2(Q-1)\pi}{Q} \right\}$, where $Q$ is set to 16. Then, the analog combining vector corresponding to the $n$th sub-antenna array is given as $\vw_n = \left[ w_{n}^{(1)}, \ldots, w_{n}^{(M)} \right]^T$, where $w_{n}^{(m)} = r_n e^{j \theta_{n}^{(m)}}$ with $r_n = \frac{1}{\sqrt{M}}$ for the ESA architecture and $r_n = \frac{1}{\sqrt{m_n}}$ for the UESA architecture. The phase of $w_{n}^{(m)}$ is selected from $\Theta$ being closest to the phase of the $m$th element of $\vu_n^{\star}$, as in step 6 of Algorithm 1. Finally, the SNR is defined as the ratio of the average symbol power of a user to noise power $N_0$.
	
	\subsection{Changing properties of the largest eigenvalues of $\mT_n$}
	\begin{figure*}[t]
		\subfigure[$N_r=32, N = K = 4$]
		{
			\includegraphics[scale=0.43]{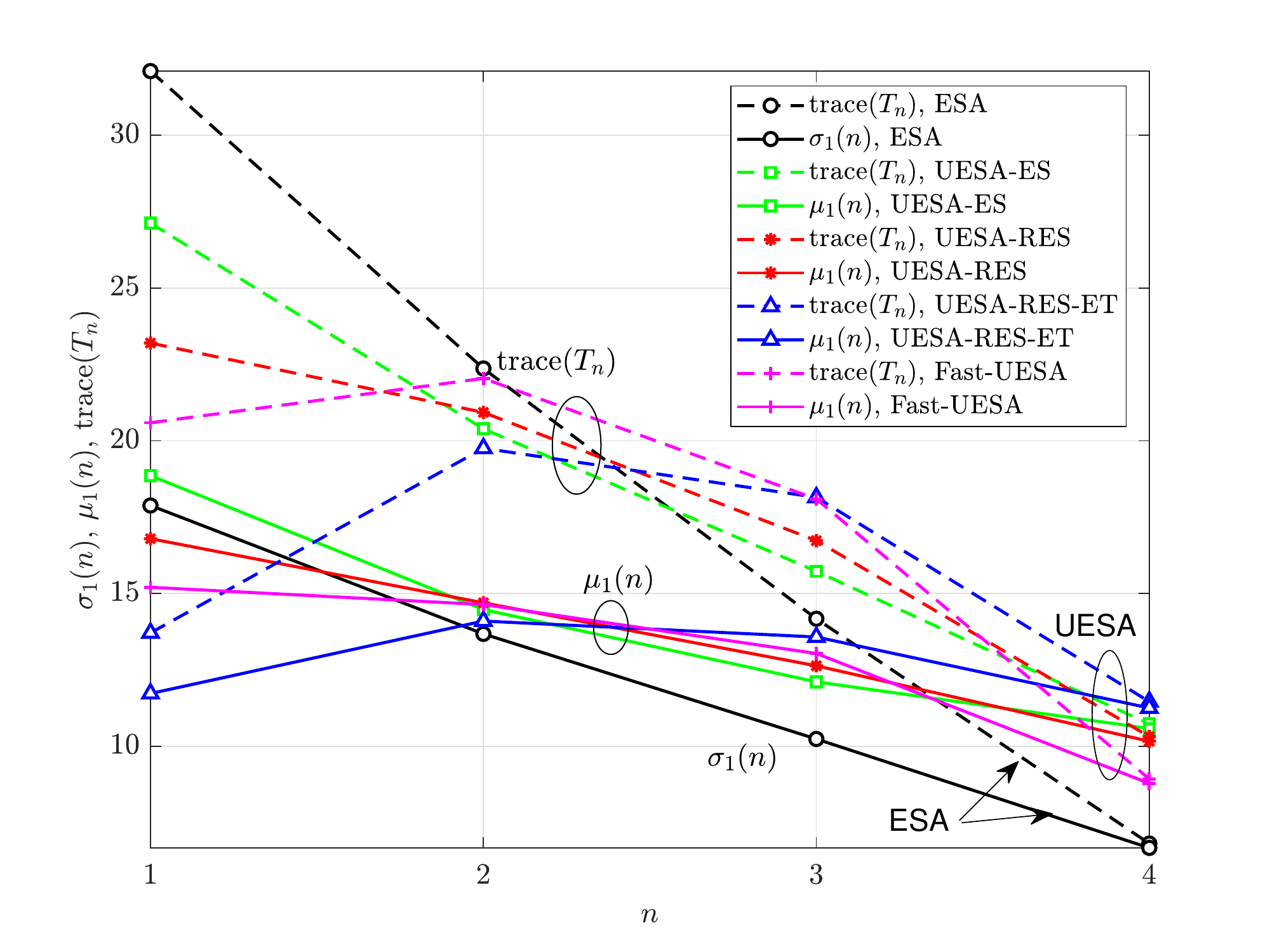}
			\label{fig:sigma_48_3_3}
		}
		\subfigure[$N_r=64, N = K = 4$]
		{
			\includegraphics[scale=0.43]{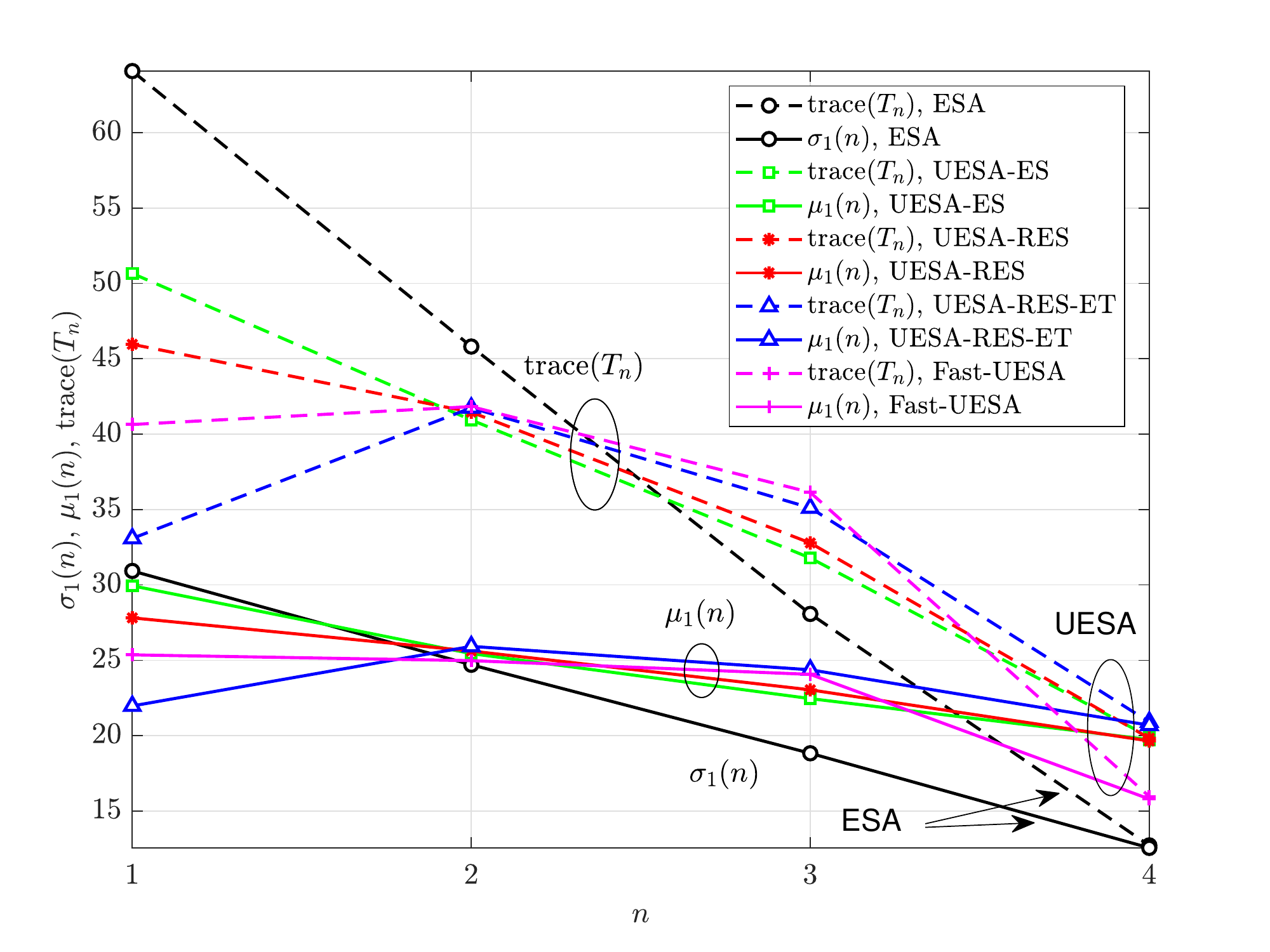}
			\label{fig:sigma_64_4_4}
		}
		\caption{Changing properties of the average largest eigenvalues of $\mT_n$ with SNR = $12$ dB.}
	\end{figure*}
	In Figs. \ref{fig:sigma_48_3_3} and \ref{fig:sigma_64_4_4}, the average largest eigenvalue of $\mT_n$ in the conventional ESA is compared with those of the proposed schemes, namely UESA-ES, UESA-RES, UESA-RES-ET, and Fast-UESA for $N=K=4$ and $N_r=\left\{32, 64\right\}$ at SNR $=12$ dB. We set $count_{max}=\left\{30, 280\right\}$ for the UESA-RES-ET scheme and $I = \left\{20, 40\right\}, \gamma = \left\{ 2, 4 \right\}$ for the Fast-UESA algorithm corresponding to $N_r=\left\{32, 64\right\}$. We also show ${\tr{\mT_n}}$, which is equal to the sum of all eigenvalues of $\mT_n$ in these schemes. From Figs. \ref{fig:sigma_48_3_3} and \ref{fig:sigma_64_4_4}, the following observations are noted:
	\begin{itemize}
		\item It is clear that in  the ESA system, similar to ${\tr{\mT_n}}$, $\sm{n}$ decreases with $n$.
		
		\item By optimizing ${\tr{\mT_n}}$ in the UESA system, the differences among $\um{1}, \um{2}, \ldots, \um{N}$ become smaller, which justifies the use of constraint \eqref{trace_equal} to optimize $\um{n}$.
		
		\item Among the UESA-ES, UESA-RES, and UESA-RES-ET algorithms,  although slightly smaller differences among $\um{1}, \um{2}, \ldots, \um{N}$ is seen in the UESA-RES-ET algorithm, UESA-ES and UESA-RES achieve higher values of $\sum_n \um{n}$, which results in their higher upper bounds of the achievable rate, as will be shown in the next subsection.
	\end{itemize}  

	\subsection{Total achievable rate of the UESA architecture}
	\begin{figure*}[t]
		% Fig. 2
		\subfigure[$N_r=32, N = K = 4$]
		{
			\includegraphics[scale=0.58]{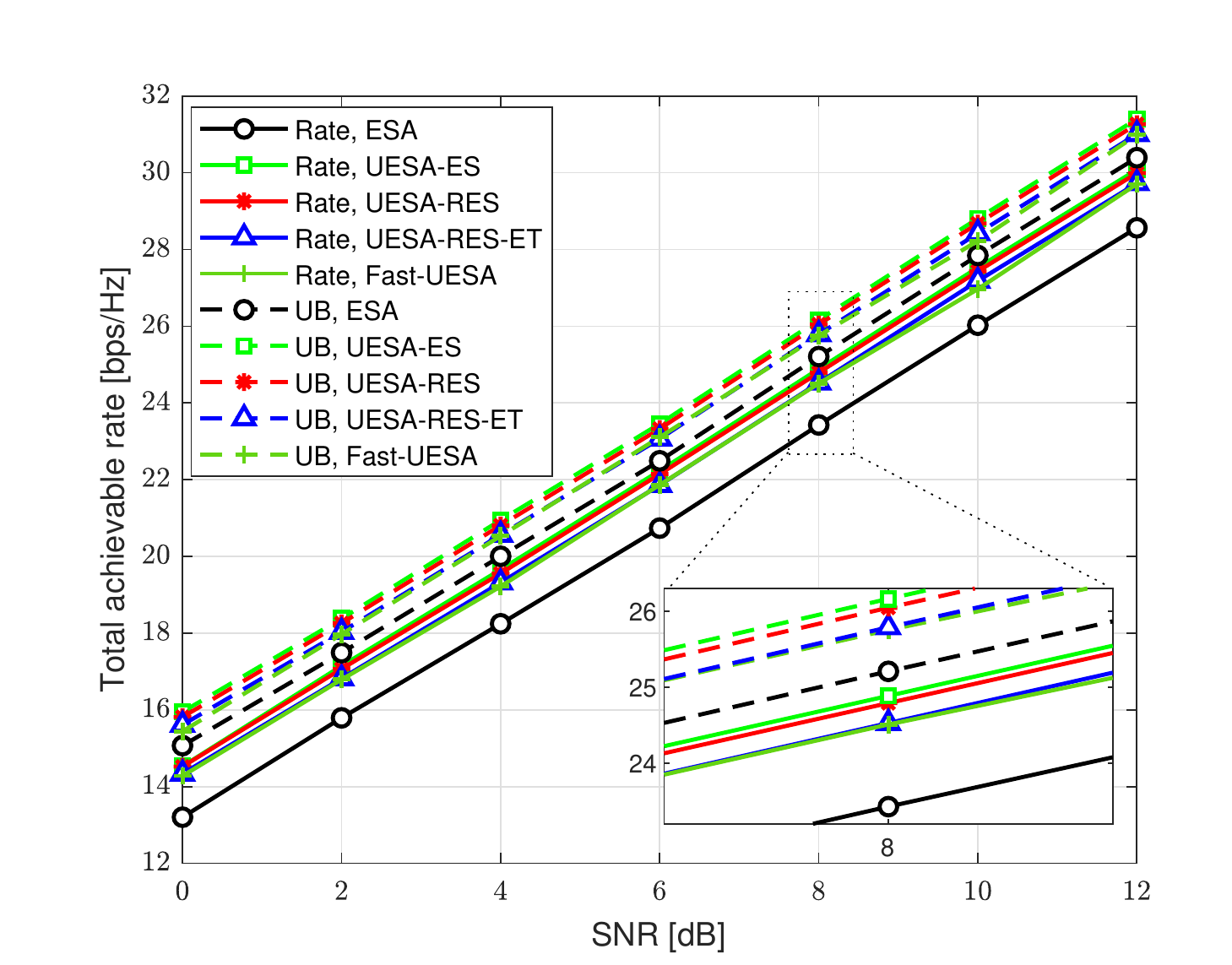}
			\label{fig:rate_48_3}
		}
		\subfigure[$N_r=64, N = K = 4$]
		{
			\includegraphics[scale=0.58]{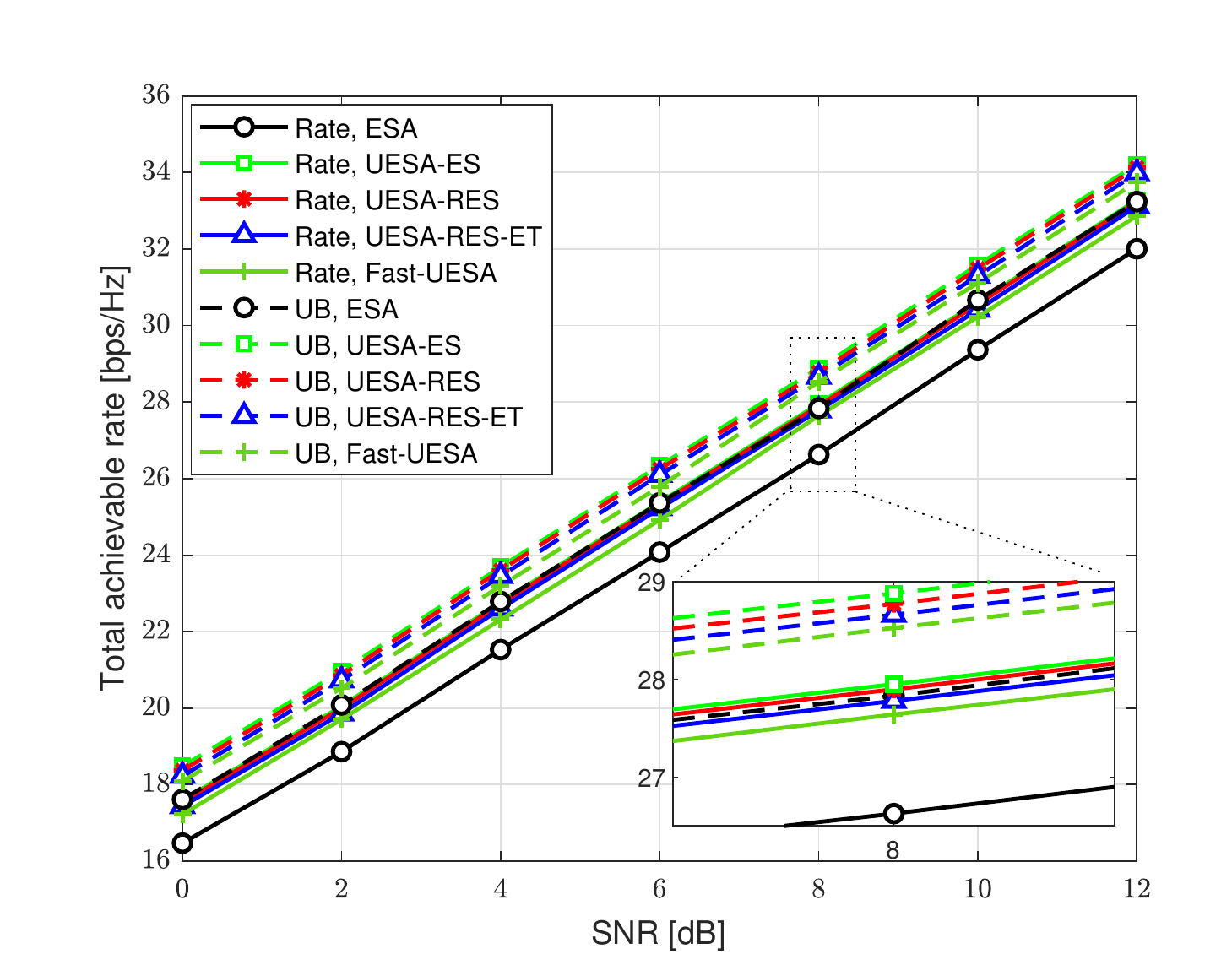}
			\label{fig:rate_64_4_4}
		}
		\caption{Comparison of total achievable rates and their upper bounds  for the conventional ESA, the proposed UESA-ES, UESA-RES, UESA-RES-ET, and Fast-UESA schemes.}
	\end{figure*}

	\begin{figure}[t]
		% Fig. 2
		\centering
		\includegraphics[scale=0.6]{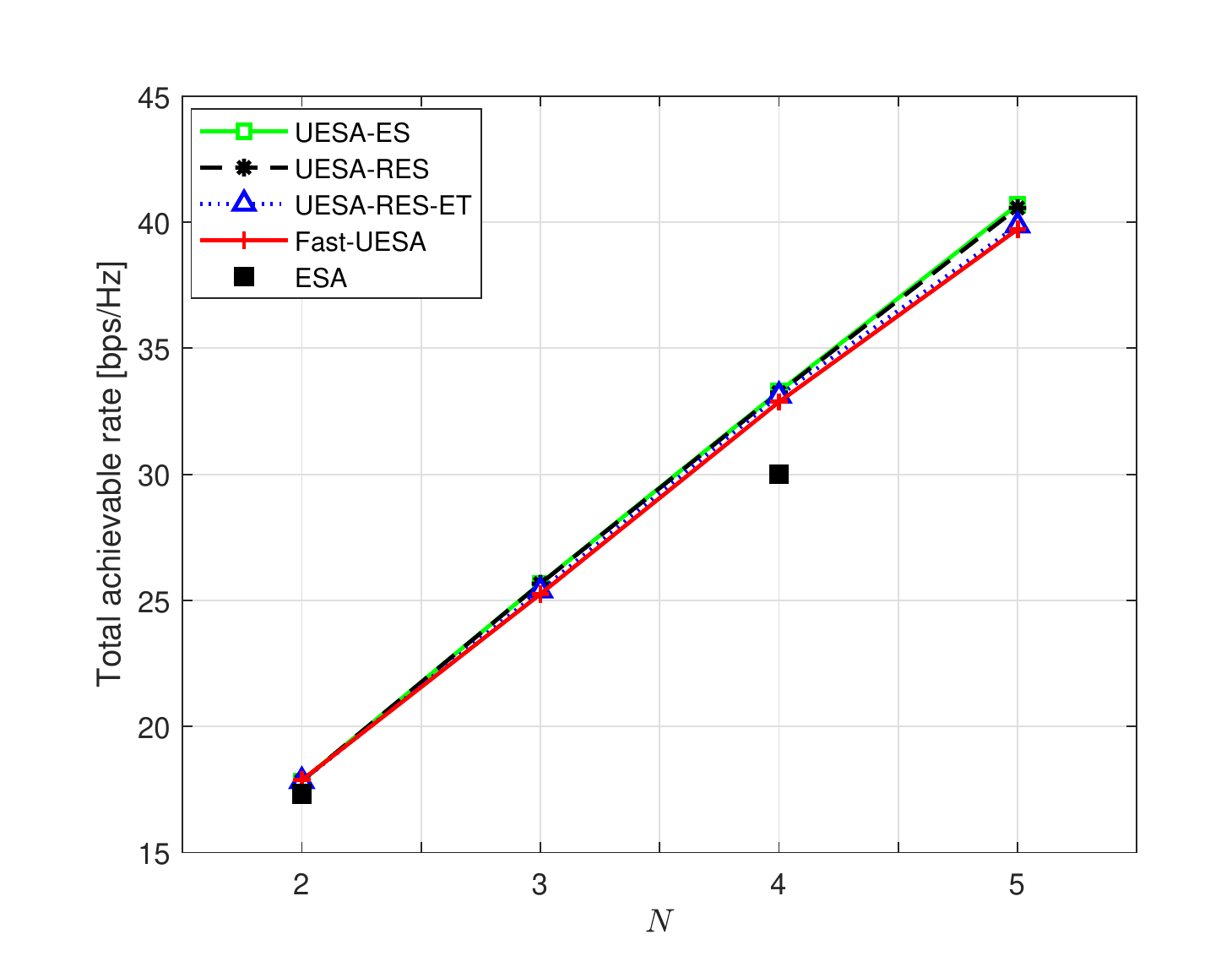}
		\caption{Comparison of total achievable rates of the conventional ESA architecture and the proposed schemes with $N_r=64$, $N=K=\left\{2,3,4,5\right\}$ and SNR $ = 12$ dB.}
		\label{fig:Rate_vs_N}
	\end{figure} 
	
	In Figs. \ref{fig:rate_48_3} and \ref{fig:rate_64_4_4}, the total achievable rates and their upper bounds (UBs) in \eqref{rate_8} and \eqref{ub_uesa}, of the ESA and UESA schemes for $N=K=4$ and $N_r=\left\{32, 64\right\}$, are presented. For $N_r=\left\{32, 64\right\}$, we set $count_{max}=\left\{30, 280\right\}$ for the UESA-RES-ET scheme and $I = \left\{20, 40\right\}, \gamma = \left\{ 2, 4 \right\}$ for the Fast-UESA scheme, respectively. In both figures, the UESA-ES algorithm provides the highest total rates and upper bounds while those of the UESA-RES, UESA-RES-ET, and Fast-UESA schemes are slightly lower. For example, with $N_r=32, N=K=4$, and SNR = 0 dB, the UESA-ES scheme achieves approximately a 10.5$\%$ higher total rate compared to the ESA scheme, while the corresponding enhancements attained by the UESA-RES, UESA-RES-ET, and Fast-UESA algorithms are approximately 10$\%$. The improvement in total rate can be explained by the upper bounds plotted in these figures. Specifically, the proposed UESA schemes not only enhance the upper bounds, but also reduce gaps between the upper bounds and achievable rates. The total achievable rates of the UESA-RES-ET  scheme are almost identical to those of the UESA-RES scheme in spite of its reduced search region. Although the Fast-UESA algorithm has much lower complexity compared with UESA-RES and UESA-RES-ET, as will be shown in Section V-D, only a marginal performance loss is seen for both systems.
	
	Fig. \ref{fig:Rate_vs_N} presents the total achievable rates of the proposed UESA and conventional ESA architectures with various numbers of RF chains $N=\left\{2, 3, 4, 5\right\}, N_r=64$, at SNR $= 12$ dB. Because the antennas are equally allocated to sub-antenna arrays in the conventional ESA architecture, it cannot serve $N=\left\{3, 5\right\}$ without significant modifications. By contrast, the proposed UESA architecture can work for every number of RF chains, which makes it more flexible for the implementation of sub-connected analog hybrid beamforming networks and made it easier to achieve a certain tradeoff between spectral and energy efficiency. Furthermore, Fig.  \ref{fig:Rate_vs_N} shows that the proposed UESA schemes outperform the ESA, and the performance improvement is clearer for large values of $N$.
	
	\subsection{Energy efficiency}
	
%	\begin{figure}[t]
%		% Fig. 2
%		\centering
%		\includegraphics[scale=0.44]{EE_vs_snr}
%		\caption{Comparison of energy efficiencies of the conventional ESA architecture and the proposed scheme with $N_r=\left\{32, 64\right\},N=K=4$.}
%		\label{fig:EE}
%	\end{figure}
%	
%	\begin{figure}[t]
%		% Fig. 2
%		\centering
%		\includegraphics[scale=0.6]{Rate_EE_vs_Nr}
%		\caption{Comparison of spectral and energy efficiencies of the conventional ESA architecture and the proposed schemes with $N_r=\left\{8, 16, 32, 48, 64\right\},N=K=4$,  and SNR $= 12$ dB.}
%		\label{fig:SE_EE_vs_Nr}
%	\end{figure} 

	\begin{figure*}[t]
		% Fig. 2
		\subfigure[$N_r=\left\{32, 64\right\},N=K=4$]
		{
			\includegraphics[scale=0.44]{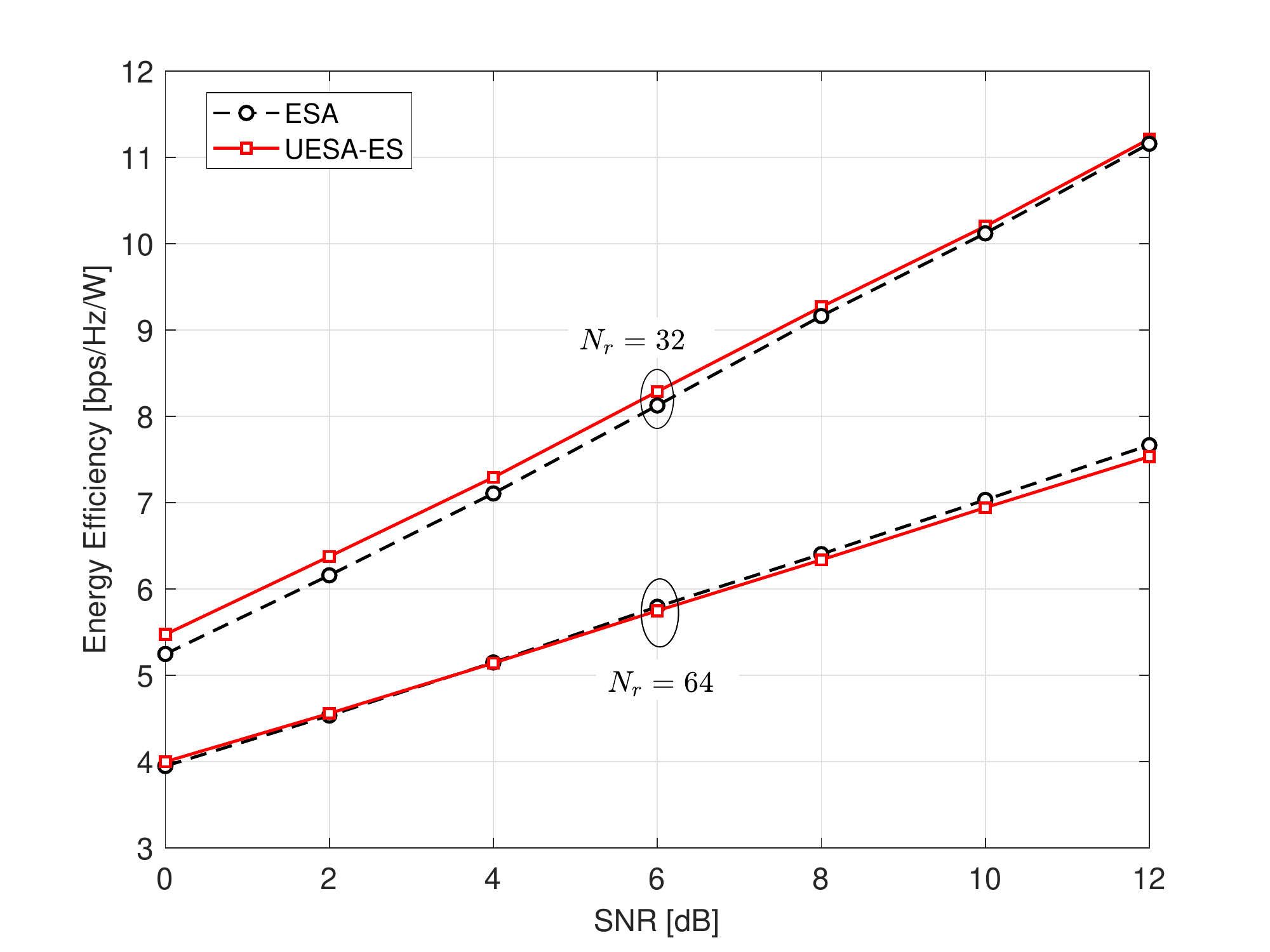}
			\label{fig:EE}
		}
		\subfigure[$N_r=\left\{8, 16, 32, 48, 64\right\},N=K=4$, SNR $= 12$ dB]
		{
			\includegraphics[scale=0.6]{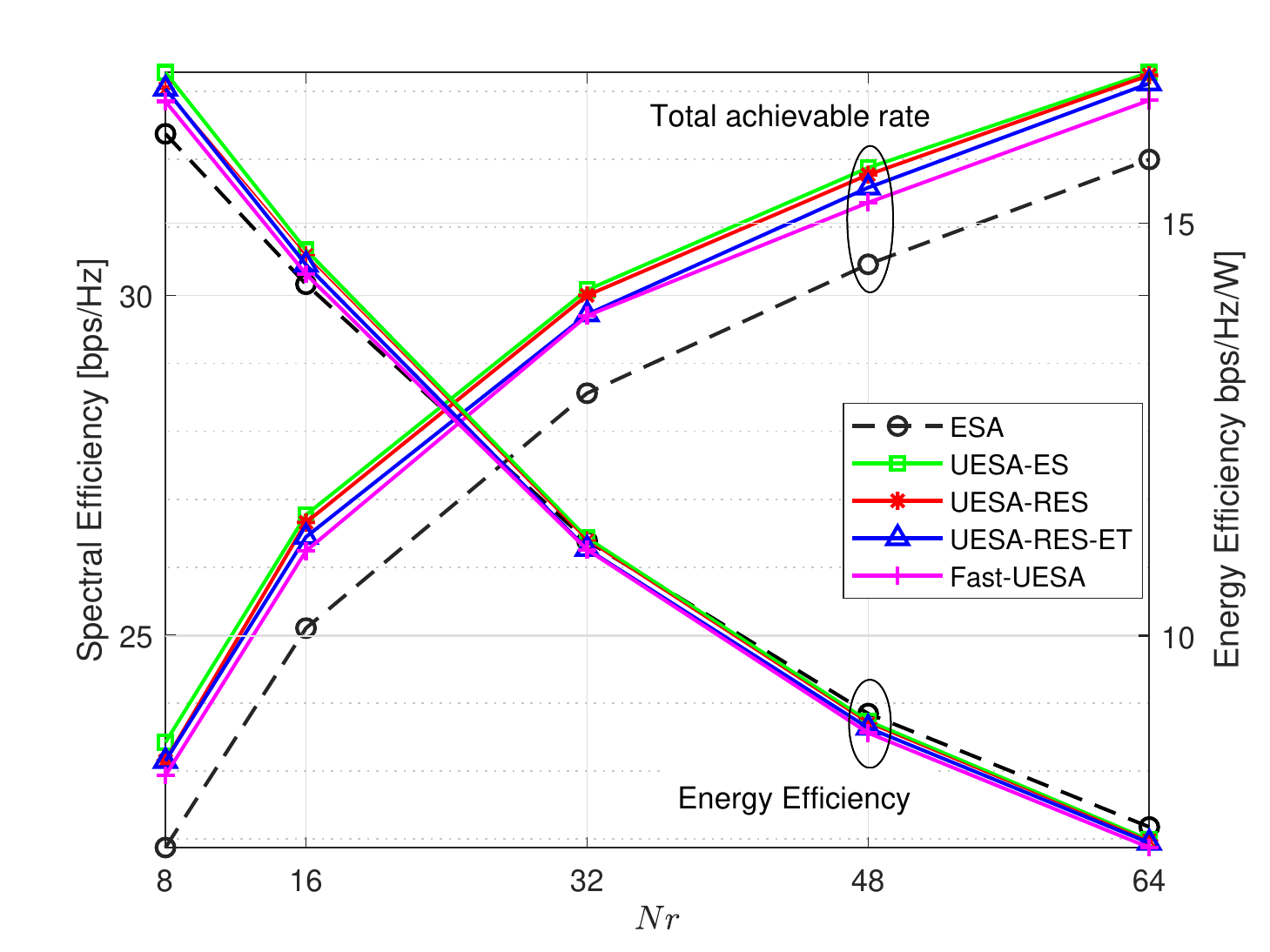}
			\label{fig:SE_EE_vs_Nr}
		}
		\caption{Comparison of spectral and energy efficiencies of the conventional ESA architecture and the proposed schemes.}
	\end{figure*}

	To compute the power consumption in \eqref{P_ESA} and \eqref{P_UESA}, we use the same assumption as in \cite{mendez2016hybrid} and \cite{mendez2015channel}, i.e., $P_{\text{LNA}} = p, P_{\text{ADC}} = 10p, P_{\text{RF}} = 2p, P_{\text{PS}} = 1.5p$, and $P_{\text{SW}} = 0.25p$, where $p=20$ mW is the reference power value. Then, from \eqref{P_ESA} and \eqref{P_UESA}, we obtain
	\begin{align*}
		P_{\text{\text{ESA}}} &= 50 N_r + 240 N \text{ (mW)}, \\
		P_{\text{\text{UESA}}} &= 54 N_r + 240 N \text{ (mW)}. 
	\end{align*} 
	The proposed UESA architecture is observed to require $4N_r$ mW of additional power. However, this is much smaller than the total required power.
	
	In Fig. \ref{fig:EE}, we compare the energy efficiency of the considered architectures for $N=K=4$ and $N_r = \left\{32, 64 \right\}$. For the UESA architecture, the UESA-ES algorithm is used. Energy efficiency is defined as the ratio of the total achievable rate to the power consumed. In both systems, we observe that even though the UESA architecture requires higher power than the ESA architecture, it achieves comparable energy efficiency to the conventional ESA for $N_r=64$. In case of $N_r = 32$, the improvement in the energy efficiency of the UESA architecture is clear, especially at low SNR values.
	
	In Fig. \ref{fig:SE_EE_vs_Nr} the total achievable rates and energy efficiencies of the proposed schemes, namely, UESA-ES, UESA-RES, UESA-RES-ET, and Fast-UESA, are compared with those of the ESA scheme for $N_r=\left\{8, 16, 32, 48, 64\right\},N=K=4$ with SNR $= 12$ dB. In this figure, the enhancement in the achievable rate of the proposed schemes is clear for all values of $N_r$. Furthermore, because the UESA architecture requires $4N_r$ mW of additional power, the gains in terms of energy efficiency of the proposed schemes decrease as $N_r$ increases. However, even at $N_r=64$, the energy efficiencies of the UESA and ESA architectures are still comparable.
	
	\subsection{Computational complexities}
	We numerically analyze the complexity reduction of the proposed near-optimal algorithms, the UESA-RES, UESA-RES-ET, and Fast-UESA, in comparison with that of the optimal UESA-ES algorithm. We note that the complexities of candidate ordering in steps 1--7 of Algorithm \ref{algorithm:UESA_RES_EE} can be done offline, and thus its complexity can be ignored. Furthermore, the complexity of ordering rows of the channel is significantly lower than that of executing Algorithm \ref{algorithm:factorization_AC} for multiple candidates. Consequently, the overall complexity becomes approximately proportional to the number of examined candidates $\comM$ in each algorithm.
	
	In Table I, we present the numbers of examined candidates in the UESA-ES, UESA-RES, UESA-RES-ET, and Fast-UESA algorithms for two environments, $N=K=4$ and $N_r=\left\{32, 64\right\}$ with SNR = $12$ dB, $count_{max} = \left\{30, 280\right\}$ for the UESA-RES-ET algorithm, and $I = \left\{20, 40\right\}, \gamma = \left\{2, 4\right\}$ for the Fast-UESA algorithm. Table I shows that the UESA-ES requires a substantially larger number of candidates than the UESA-RES, UESA-RES-ET, and Fast-UESA. By contrast, the UESA-RES and UESA-RES-ET algorithms explore the reduced search region, and thus significantly reduce complexity while producing only marginal performance losses, as observed in Fig. 4. Furthermore, the Fast-UESA algorithm has significantly lower complexity compared with the other proposed algorithms. Specifically, in a large MIMO system with $N_r=64$ antennas and $N=4$ RF chains, the UESA-ES tests approximately $4\times 10^4$ candidates to find the optimal solution, but the proposed UESA-RES and UESA-RES-ET algorithms need to test only $1906$ and $770$ candidates, approximately $5\%$ and $2\%$, respectively, of that tested by the UESA-ES. Meanwhile, the Fast-UESA algorithm requires testing only 35 candidates, corresponding to $0.9 \%, 1.8 \%, $ and $4.5 \%$ complexity of the UESA-ES, UESA-RES, and UESA-RES-ET algorithm.
	\begin{table}[t]
		\renewcommand{\arraystretch}{1.3}
		\caption{Comparison of the average numbers of candidates examined by the\\ UESA-ES, UESA-RES, UESA-RES-ET, and Fast-UESA algorithms for $N=K=4$ and $N_r=\left\{32, 64\right\}$ at SNR = $12$ dB}
		\label{tab:Comparison_mul}
		\centering
		\begin{tabular}{c|c|c|c}
			\hline
			
			\multicolumn{2}{c|}{{Algorithms}}  &  \makecell{$N_r=32$}    & \makecell{$N_r=64$}\\
			
			\hline
			\hline
			
			\multicolumn{2}{c|}{UESA-ES}    &  $4495$  &    $39711$\\
			\hline
			
			\multicolumn{2}{c|}{UESA-RES}    &  $249$  &   $1906$\\
			\hline
			
			\multicolumn{2}{c|}{UESA-RES-ET}    &  \makecell{$72$ $(count_{max}=30)$}  &   \makecell{$770$ $(count_{max}=280)$}\\
			\hline
			
			\multicolumn{2}{c|}{Fast-UESA}    &  \makecell{$28$ $(I=20)$}  &   \makecell{$35$ $(I=40)$}\\
			\hline
		\end{tabular}
	\end{table}

	\subsection{Performance comparison between the ESA, UESA, AS, AHB, and combined UESA-AHB schemes}
	
	\begin{figure}[t]
		% Fig. 2
		\centering
		\includegraphics[scale=0.6]{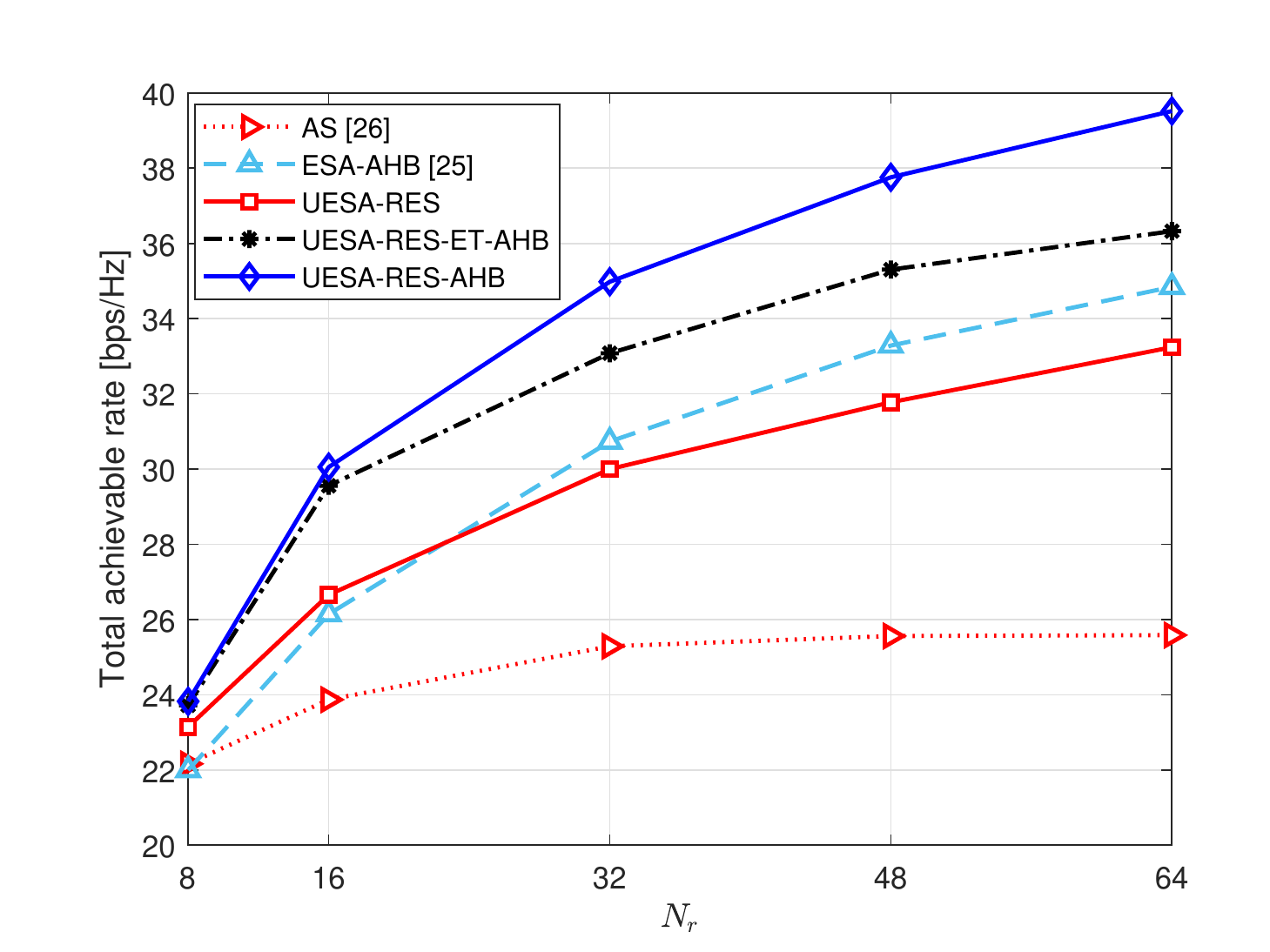}
		\caption{Comparison of total achievable rates of the AS \cite{chen2007efficient}, ESA-AHB \cite{zhu2016adaptive}, and the combined UESA-AHB schemes, namely, UESA-RES-AHB and UESA-RES-ET-AHB, with $N_r = \left\{8,16,32,48,64\right\}, N=K=4$, and SNR $= 12$ dB.}
		\label{fig:Rate_vs_Nr_comp_AHP}
	\end{figure}

	In Fig. \ref{fig:Rate_vs_Nr_comp_AHP}, we compare the performance of the proposed UESA scheme to those of the antenna selection (AS) \cite{chen2007efficient} and the adaptive hybrid beamforming (AHB) \cite{zhu2016adaptive} schemes with $N_r = \left\{8,16,32,48,64\right\}$, $N=K=4$, and SNR $=12$ dB. In \cite{zhu2016adaptive}, the AHB scheme is proposed for the conventional ESA architecture, and we refer to it as ESA-AHB to avoid confusion. For the UESA architecture, the UESA-RES algorithm is used. We observe that while the ESA-AHB scheme performs better than the UESA-RES scheme for $N_r \geq 32$, the UESA-RES scheme achieves better performance for $N_r \leq 16$, whereas the AS scheme performs far worse than both the UESA-RES and ESA-AHB schemes for all the considered systems.

	We note that the AHB algorithm can be combined with the UESA architecture to create an unequal sub-array architecture with adaptive hybrid beamforming (UESA-AHB). Specifically, by inheriting the design of the UESA and AHB schemes, the UESA-AHB scheme employs a combining matrix that has different numbers of non-zero elements in combining vectors reflecting the unequal numbers of antennas in sub-antenna arrays, and the positions of non-zero elements are distributed over the range $[1, N_r]$. Furthermore, unlike the ESA-AHB scheme, in the UESA-AHB scheme, the combining matrix depends on different numbers of assigned antennas across sub-arrays. Therefore, the combining matrix in the UESA-AHB scheme should be jointly optimized with the antenna allocation of the UESA architecture. This can be done by applying the AHB scheme to step 5 of Algorithm \ref{algorithm:RES} and step 13 of Algorithm \ref{algorithm:UESA_RES_EE}, resulting in the UESA-RES-AHB and the UESA-RES-ET-AHB schemes, respectively.
	
	In Fig. \ref{fig:Rate_vs_Nr_comp_AHP}, we show the total achievable rates of the UESA-RES-AHB and UESA-RES-ET-AHB schemes; and $count_{max}=40$ is used for early termination in the UESA-RES-ET-AHB scheme. It is clearly seen that the combined UESA-AHB schemes achieve substantial gains over the existing schemes for all the considered values of $N_r$. In particular, the UESA-RES-AHB scheme shows an approximately $14.8 \%$ improvement in the total achievable rate with respect to the ESA-AHB scheme with $N_r=64$.
		
	We note that the same switching network is employed by both the UESA architecture and the ESA-AHB scheme \cite{zhu2016adaptive}. Specifically, the switching network in Fig. \ref{fig:DynamicUESA} can be used for both antenna allocation in the UESA architecture and RF chain-to-antenna connection in the ESA-AHB scheme. Therefore, the power consumptions of the UESA-RES, ESA-AHB, and UESA-AHB schemes are the same. As a result, with the significant improvement in the total achievable rate, the UESA-AHB schemes achieve the highest energy efficiency compared with both the UESA-RES and UESA-AHB schemes.
	
	\section{Conclusion}
	
	In this work, we propose a novel unequal sub-connected architecture for a hybrid combining at the receiver of massive MIMO systems. Unlike the conventional ESA architecture, the proposed UESA architecture assigns different numbers of antennas to sub-antenna arrays based on channel conditions. Our analytical derivations suggest that when the factorization-aided analog-combining algorithm is employed, fewer antennas should be assigned to the first sub-antenna arrays to approach the upper bound of the total achievable rate, while more antennas should be assigned to the last sub-antenna arrays. We also show that channel rows should be ordered in decreasing order of their norms to enhance the upper bound of the total achievable rate. The simulation results demonstrate that the proposed architecture can achieve up to a $10\%$ improvement in total rate with a small increase in power consumption. To reduce the complexity of determining the antenna-to-sub-array connections, the near-optimal UESA-RES, UESA-RES-ET, and Fast-UESA algorithms are proposed. The numerical results show that they can significantly reduce the complexity of the UESA-ES with marginal performance losses. The proposed UESA architecture requires lower power consumption while improving the spectral efficiency of the sub-connected architecture. These advantages can also be beneficial for the transmitter of the downlink. Therefore, in future research, the spectral-efficiency analysis and antenna allocation algorithms for the UESA architecture can be applied to the transmitter of the downlink MIMO systems.
	
	\appendices
	
	\section{Proof of Lemma \ref{lemma:factorization_rate}}
	\label{appendix:factorization_rate}
	For simplicity of presentation, we define $\mQ = \mI_K+  \rho \mH^H \mW \mW^H \mH$, and then get
	\begin{align*}
		\mQ = \mI_K+  \rho \mG_1 + \rho \mG_2 + \ldots + \rho \mG_N. \numberthis \label{Q_1}
	\end{align*}
	By defining $\mE_1 = \mI_K+  \rho \mG_1$, \eqref{Q_1} can be rewritten as $\mQ  = \mE_1 \left(\mI_K + \rho \mE_1^{-1} \mG_2 + \ldots + \rho \mE_1^{-1} \mG_N\right).$ Similarly, by defining $\mE_2 = \mI_K + \rho \mE_1^{-1} \mG_2$, we obtain
	\begin{align*}
		\mQ 
		= \mE_1 \mE_2 \left(\mI_K + \rho \mE_2^{-1} \mE_1^{-1} \mG_2 + \ldots + \rho \mE_2^{-1} \mE_1^{-1} \mG_N\right).
	\end{align*}
	In a similar manner, $\mQ$ can be factorized to the product of $\mE_n, n=1,\ldots,N$, i.e.,
	\begin{align*}
		\mQ 
		= \mE_1 \mE_2 \ldots \mE_N, \numberthis \label{eqn_Q}
	\end{align*}
	where $\mE_n = \mI_K + \rho (\mE_1 \ldots \mE_{n-1})^{-1} \mG_n$. By defining $\mQ_0 = \mI_K$ and $\mQ_{n} = \mE_1 \ldots \mE_{n}$, $n = 1, \ldots, N-1$, we have
	\begin{align*}
		\mQ_n = \mQ_{n-1} \mE_n, \numberthis \label{Q_i}
	\end{align*}
	and $\mE_n$ can be expressed as
	\begin{align}
		\label{eqn_E}
		\mE_n = \mI_K + \rho \mQ_{n-1}^{-1} \mG_n = \mI_K + \rho \mQ_{n-1}^{-1} \mH_n^H \vw_n \vw_n^H \mH_n.
	\end{align}
	From \eqref{eqn_Q} and \eqref{eqn_E}, \eqref{rate_2} can be rewritten as
	\begin{align*}
		R &= \log_2 \det \mQ = \log_2 \det \left( \mE_1 \mE_2 \ldots \mE_N \right)
		=  \sum_{n=1}^{N} \log_2 \det \left(\mI_K + \rho \mQ_{n-1}^{-1} \mH_n^H \vw_n \vw_n^H \mH_n \right) \\
		&=  \sum_{n=1}^{N} \log_2 \left(1 + \rho \vw_n^H \mH_n \mQ_{n-1}^{-1} \mH_n^H \vw_n \right), \numberthis \label{rate_3}
	\end{align*}
	where the last equality follows the fact that $\det(\mI_K + \va \vb^T) = 1 + \vb^T \va$ with $\va = \mQ_{n-1}^{-1} \mH_n^H \vw_n$ and $\vb^T = \vw_n^H \mH_n$. For simplicity, we define $\mT_n = \mH_n \mQ_{n-1}^{-1} \mH_n^H$, by inserting \eqref{eqn_E} into \eqref{Q_i}, we get
	\begin{align}
		\label{Q_n_11}
		\mQ_{n-1} = \mQ_{n-2} + \rho \mG_{n-1}.
	\end{align}
	Then, from \eqref{Q_n_11}, \eqref{rate_3} can be written as \eqref{rate_4}. This completes the proof.
	
	\section{Proof of Lemma \ref{lemma:trace_Tn_decrease}}
	\label{appendix:prove_traceTn_decrease}
	We have
	\begin{align*}
		{\tr{\mT_n}} 
		= {\tr{\mH_n \mQ_{n-1}^{-1} \mH_n^H}}
		= \tr{{\mH_n^H \mH_n}  {\mQ_{n-1}^{-1}} }. \numberthis \label{tr_2}
	\end{align*}
	When $N_r$ grows while $K$ is kept constant, we have $\mH_n^H \mH_n \approx M \mI_K$ \cite{lu2014overview}, which yields
	\begin{align*}
		{\tr{\mT_n}} 
		\approx M {\tr{\mQ_{n-1}^{-1}}}. \numberthis \label{tr_3}
	\end{align*}
	Because we have $\mQ_{n-1}  = \mQ_{n-2} \mE_{n-1}$, $\tr{\mQ_{n-1}^{-1}}$ in \eqref{tr_3} can be expressed as
	\begin{align*}
		\tr{\mQ^{-1}_{n-1}} 
		&= \tr{\mE_{n-1}^{-1} \mQ^{-1}_{n-2}}
		\leq   \lambda_{max}\left(\mE_{n-1}^{-1} \right) \tr{\mQ^{-1}_{n-2}} \numberthis \label{tr_33}\\
		&= \lambda_{min}^{-1}\left(\mE_{n-1} \right) \tr{\mQ^{-1}_{n-2}}, \numberthis \label{tr_4}
	\end{align*}
	where $\lambda_{max}(\mA)$ and $\lambda_{min}(\mA)$ are the largest and smallest eigenvalues of  $\mA$, respectively. Here, \eqref{tr_33} is obtained by the inequality $\tr{\mA \mB} \leq \lambda_{max}(\mA) \tr{\mB}$ \cite{kleinman1968design, wang1986trace}. From \eqref{eqn_E}, we get
	\begin{align*}
		\lambda_{min} \left(\mE_{n-1} \right) = 1 + \rho \lambda_{min} \left(\mQ_{n-1}^{-1} \mG_n\right). \numberthis \label{lamda_min}
	\end{align*}
	Based on $\mG_{n-1} = \mH_{n-1}^H \vw_{n-1} \vw_{n-1}^H \mH_{n-1}$, $\mQ_0 = \mI_N$, and \eqref{Q_n_11}, we find that $\mQ_{n-1}^{-1} \mG_n$ in \eqref{lamda_min} is a semidefinite Hermitian matrix, and $\lambda_{min} \left(\mQ_{n-1}^{-1} \mG_n\right) > 0$. As a result, we have $\lambda_{min} \left(\mE_{n-1} \right) > 1$ and $\lambda_{min}^{-1}\left(\mE_{n-1} \right) < 1$. Then, from \eqref{tr_4}, we obtain
	\begin{align*}
		\tr{\mQ^{-1}_{n-1}} <  \tr{\mQ^{-1}_{n-2}}, \numberthis \label{tr_5}
	\end{align*}
	which implies that $\tr{\mQ_{n}^{-1}}$ is a decreasing function of $n$. In the conventional ESA system, $M$ in \eqref{tr_3} is a constant. Consequently, ${\tr{\mT_n}}$ decreases with $n$.
	
	\section{Proof of Theorem 1}
	\label{appendix:proof_of_theorem1}
	
	Considering that $\mG_{n-1} = \mH_{n-1}^H \vw_{n-1} \vw_{n-1}^H \mH_{n-1}$ is of rank one, and according to \eqref{Q_n_11} and the result in \cite{miller1981inverse}, we obtain
	\begin{align*}
		\mQ_{n-1}^{-1} = \left(\mQ_{n-2} + \rho \mG_{n-1}\right)^{-1} 
		= \mQ_{n-2}^{-1} - \frac{\rho \mQ_{n-2}^{-1} \mG_{n-1} \mQ_{n-2}^{-1}}{1 + \rho \tr{\mG_{n-1} \mQ_{n-2}^{-1}}}. \numberthis \label{Q_inv}
	\end{align*}
	Then, the difference between $\tr{\mQ_{n-1}^{-1}}$ and $\tr{\mQ_{n-2}^{-1}}$ is given by
	\begin{align*}
		\Delta_{n-1} 
		= {\tr{\mQ_{n-2}^{-1}}} -{\tr{\mQ_{n-1}^{-1}}} 
		= \frac{\rho \tr{ {\mQ_{n-2}^{-1}} {\mG_{n-1}} {\mQ_{n-2}^{-1}}}}{ 1 + \rho \tr{{\mG_{n-1}} {\mQ_{n-2}^{-1}}}}. \numberthis \label{Delta_1}
	\end{align*}
 	Furthermore, $\mQ_{n-2}$ can be expressed as
	\begin{align*}
		\mQ_{n-2} = \mI_K + \rho \mG_1 + \ldots + \rho \mG_{n-2} = \mI_K + \rho \tilde{\mG}_{n-2},
	\end{align*}
	where $\tilde{\mG}_{n-2} = \rho \mG_1 + \ldots + \rho \mG_{n-2} $. According to the result in \cite{henderson1981deriving}, we have
	\begin{align*}
		\mQ_{n-2}^{-1} 
		= \mI_K - \rho \tilde{\mG}_{n-2} \left(\mI_K + \rho \tilde{\mG}_{n-2}\right)^{-1}
		= \mI_K - \rho \tilde{\mG}_{n-2} \mQ_{n-2}^{-1} ,
	\end{align*}
	which yields
	\begin{align*}
		{\mQ_{n-2}^{-1}} {\mG_{n-1}} {\mQ_{n-2}^{-1}} 
		= \mQ_{n-2}^{-1} \mG_{n-1} - \rho \mQ_{n-2}^{-1} \mG_{n-1} \tilde{\mG}_{n-2} \mQ_{n-2}^{-1}. \numberthis \label{num_1}
	\end{align*}
	In the second term on the right-hand side of \eqref{num_1}, we have $\mG_{n-1} \tilde{\mG}_{n-2} = \rho \sum_{i=1}^{n-2} \mG_{n-1} \mG_i$. Therefore, \eqref{num_1} yields
		\begin{align*}
			\tr {{\mQ_{n-2}^{-1}} {\mG_{n-1}} {\mQ_{n-2}^{-1}}} 
			= \tr {\mQ_{n-2}^{-1} \mG_{n-1}} - \rho^2 \sum_{i=1}^{n-2} \underbrace{\tr {\mQ_{n-2}^{-1}  \mG_{n-1} \mG_i \mQ_{n-2}^{-1}}}_{\triangleq \Upsilon}. \numberthis \label{num_3}
		\end{align*}
		By using the inequality $\tr{\mA \mB} \geq \lambda_{min}(\mA) \tr{\mB}$ \cite{kleinman1968design, wang1986trace}, we have
		\begin{align*}
			\Upsilon 
			&\geq \lambda_{min}(\mQ_{n-2}^{-1}  \mG_{n-1}) \tr{\mG_i \mQ_{n-2}^{-1}} 
			\geq \lambda_{min}(\mQ_{n-2}^{-1}  \mG_{n-1}) \lambda_{min}(\mG_i)  \tr{ \mQ_{n-2}^{-1}} 
			\geq 0, \numberthis \label{num_4}
		\end{align*}
		where the inequality in \eqref{num_4} is obtained because $\mQ_{n-2}^{-1}$ and $\mG_{n-1}$ are positive semidefinite Hermitian  matrices for $n=1,2,\ldots,N$. From \eqref{num_3} and \eqref{num_4}, we have
		\begin{align*}
			\tr {{\mQ_{n-2}^{-1}} {\mG_{n-1}} {\mQ_{n-2}^{-1}}} \leq \tr {\mQ_{n-2}^{-1} \mG_{n-1}}. \numberthis \label{num_5}
		\end{align*}
	From \eqref{Delta_1} and \eqref{num_5}, we have $\Delta_{n-1} < 1$, which leads to $\sum_{n=1}^{N-1} \Delta_{n-1} < N-1$, and hence,
	\begin{align*}
		\tr{\mQ_{N-1}^{-1}} &= \tr{\mQ_{0}^{-1}} - \sum_{n=1}^{N-1} \Delta_{n-1} 
		> N - (N-1) = 1 \numberthis \label{tr_QN1}
	\end{align*}
	with the note that $\mQ_{0} = \mI_N$.

	Now assume that the UESA architecture is designed such that 
	\begin{align*}
		m_1 {\tr{\mQ_{0}^{-1} }} \approx \ldots \approx  m_N {\tr{\mQ_{N-1}^{-1} }}, \numberthis \label{m_trace}
	\end{align*}
	which yields $$m_n \approx \frac{m_1 \tr{\mQ_0^{-1}}}{\tr{\mQ_{n-1}^{-1}}} = \frac{N m_1 }{\tr{\mQ_{n-1}^{-1}}},$$ where the second equality is obtained from the fact that $\mQ_0 = \mI_N$. Furthermore, due to $\sum_{n=1}^{N} m_n = N_r$ and the decreasing property of $\tr{\mQ_n^{-1}}, n=0,\ldots,N-1$, as proven in Appendix \ref{appendix:prove_traceTn_decrease}, we have $N_r \approx N \sum_{n=1}^{N} \frac{m_1 }{\tr{\mQ_{n-1}^{-1}}} <  \frac{N^2 m_1 }{\tr{\mQ_{N-1}^{-1}}},$ which approximately leads to $m_1 > \frac{N_r \tr{\mQ_{N-1}^{-1}}}{N^2}$. From \eqref{tr_QN1}, we have $m_1 > \frac{N_r}{N^2}$. In hybrid beamforming for massive MIMO systems, it is generally assumed that a relatively small number of RF chains is used, i.e., $N \ll N_r$. Therefore, when $N_r$ grows and $N$ is kept constant, $m_1$ also grows. Considering that $m_1 \leq m_2 \leq \ldots \leq m_N$, we can conclude that for the proposed design of the UESA architecture, we have $\mH_n^H \mH_n \approx m_n \mI_K, \forall n \leq N$ in massive MIMO systems. From \eqref{tr_2}, we can write
	\begin{align*}
		{\tr{\mT_n}} 
		\approx m_n {\tr{\mQ_{n-1}^{-1}}}. \numberthis \label{tr_6}
	\end{align*}
	From \eqref{m_trace} and \eqref{tr_6}, we obtain \eqref{trace_equal} as the designing objective $\mathcal{O}_1$.
	
	\section{Rate of decrease of ${\tr{\mQ_{n-1}^{-1}}}$}
	\label{appendix:decreasing_rate_with_channel_ordering}
	
	In the second term on the right-hand side of \eqref{num_3}, we have
	\begin{align*}
		\mG_{n-1} \mG_i &= \mH_{n-1}^H \vw_{n-1} \vw_{n-1}^H  \mH_{n-1} \mH_{i}^H \vw_{i} \vw_{i}^H \mH_{i}. \numberthis \label{num_6}
	\end{align*}
	In Appendix \ref{appendix:proof_of_theorem1}, it is proven that in massive MIMO systems employing the UESA architecture, we have $\mH_n^H \mH_n \approx m_n \mI_K, \forall n \leq N$. Therefore, in \eqref{num_6}, we have $\mH_{n-1} \mH_i^H \approx \mathbf{0}, i \neq n-1,$ in massive MIMO systems \cite{lu2014overview}, which yields $\mG_{n-1} \mG_i \approx \mathbf{0}$. Consequently, in \eqref{num_3} we have $\Upsilon \approx 0$.
	
	Now, the difference between $\tr{\mQ_{n-1}^{-1}}$ and $\tr{\mQ_{n-2}^{-1}}$ in \eqref{Delta_1} can be approximated by
	\begin{align*}
		\Delta_{n-1} 
		&\approx \frac{\rho \tr{ {\mG_{n-1}} {\mQ_{n-2}^{-1}}}}{ 1 + \rho \tr{{\mG_{n-1}} {\mQ_{n-2}^{-1}}}}. \numberthis \label{Delta_2}
	\end{align*}
	Furthermore, we have
	\begin{align*}
		\sm{n} 
		&= {\vu_n^{\star}}^H \mU \Sigma \mU^H \vu_n^{\star} 
		= \tr{{\vu_n^{\star}}^H \mH_n \mQ_{n-1}^{-1} \mH_n^H \vu_n^{\star}}
		= \tr{\mH_n^H \vu_n^{\star} {\vu_n^{\star}}^H \mH_n \mQ_{n-1}^{-1}}. %\numberthis \label{sm_appr}
	\end{align*}
	Based on step 8 in Algorithm \ref{algorithm:factorization_AC}, we have $\tr{\mG_n \mQ_{n-1}^{-1}} = \tr{\mH_n^H \vw_n^{\star} {\vw_n^{\star}}^H \mH_n \mQ_{n-1}^{-1}}$ 
	with $\vw_n^{\star} = \mathcal{Q} \left(\vu_n^{\star}\right)$. Therefore, when $m_n$ grows in massive MIMO systems, we have $\tr{\mG_n \mQ_{n-1}^{-1}}$ $\rightarrow \infty$ as $\sm{n} \rightarrow \infty$. Then, $\Delta_{n-1}$ approaches one for a fixed SNR. Thus, channel ordering leads to no significant difference in the rate of decrease of $\tr{\mQ_{n-1}^{-1}}$.
	
	%% References
	\bibliographystyle{IEEEtran}
	\bibliography{IEEEabrv,Reference}
	
\end{document}